\documentclass{article}
\usepackage[a4paper, total={6.5in, 10in}]{geometry}
\usepackage[utf8]{inputenc}
\usepackage[dvipsnames]{xcolor}
\usepackage{amsfonts}
\usepackage{amsthm}
\usepackage{amsmath}
\usepackage{parskip}
\usepackage{mathrsfs}
\usepackage{amssymb}
\usepackage{xtab}
\usepackage{setspace}
\usepackage{datetime}
\usepackage{svg}
\usepackage{csquotes}
\usepackage{fancyhdr}
\usepackage{float}
\usepackage{comment}
\usepackage{enumitem}
\usepackage{mathalfa}
\usepackage{array}
\usepackage{titlesec}
\usepackage{mdframed}
\usepackage{listings}
\usepackage{booktabs}
\usepackage{dsfont}
\usepackage{textcomp}
\usepackage{graphicx}
\usepackage{wrapfig}
\usepackage{cases}

\usepackage[font=small,labelfont=bf]{caption}

\usepackage[utf8]{inputenc}
\usepackage[english]{babel}

\usepackage{placeins}
\usepackage{subcaption}
\usepackage{hyperref}
\usepackage{cleveref}

\usepackage{cite}

\usepackage{multirow}

\graphicspath{{Figures/}}

\allowdisplaybreaks

\usepackage{thmtools}
\usepackage{thm-restate}
\declaretheorem[name=Proposition,numberwithin=section]{proposition}

\declaretheorem[name=Theorem,numberwithin=section]{theorem}
\declaretheorem[name=Definition,numberwithin=section]{definition}
\declaretheorem[name=Remark,numberwithin=section]{remark}
\declaretheorem[name=Corollary,numberwithin=section]{corollary}
\declaretheorem[name=Assumption]{assumption}

\newcounter{subassumption}

\usepackage{circuitikz}
\usetikzlibrary{decorations.pathreplacing,calligraphy}

\newcommand{\EE}{\mathds{E}}
\newcommand{\var}{\text{Var}}
\usepackage{lscape}
\usepackage{authblk}
\usepackage{tabularray}

\title{Emergent structures in coupled opinion and network dynamics}
\author[1,2]{Andrew Nugent}
\author[3]{Carmen Calatayud Fern\'andez}
\author[2]{Susana N. Gomes}
\affil[1]{Department of Mathematics, University College London}
\affil[2]{Mathematics Institute, University of Warwick}
\affil[3]{Department of Physics, University of Warwick}

\begin{document}

\maketitle

\begin{abstract}
    This paper investigates a model of opinion formation on an adaptive social network, consisting of a system of coupled ordinary differential equations for individuals' opinions and corresponding network edge weights. A key driver of the system's behaviour is the form of the interaction function, which determines the strength of interactions based on the distance between individuals' opinions and appears in both opinion and network dynamics. Two cases are examined: in the first the interaction function is always positive and in the second case the interaction function is of bounded-confidence type. In both cases there is positive feedback between opinion clustering and the emergence of community structure in the social network. This is confirmed through analytical results on long-term behaviour, extending existing results for a fixed network, as well as through numerical simulations. Transient network dynamics are also examined through a short-time approximation that captures the `typical' early network dynamics. Each approach improves some aspect of our understanding of the interplay between opinion and network evolution. 
\end{abstract}

\tableofcontents

\newpage

\section{Introduction} \label{Section: Introduction}

From initially simple models of group decision making \cite{degroot1974reaching} and binary opinion switching \cite{liggett1985interacting}, through the bounded the confidence models of Hegselmann-Krause (HK) \cite{hegselmann2002opinion} and Deffuant-Weisbuch (DW) \cite{deffuant2000mixing}, to the rapidly growing field of modern opinion dynamics, models have become increasingly complex and varied. Despite this, common themes emerge, such as the formation of distinct opinion clusters and a transition from consensus to polarisation. 

As models of opinion formation on adaptive networks become increasingly common, another theme can be seen: the positive reinforcement between opinion clustering and community structure in networks. Just as opinion formation models vary, with opinions taking either discrete or continuous values, updates that may be random or deterministic, synchronous or asynchronous, multiple approaches have been proposed to describe adaptive social networks. In those based on the DW model or similar, in which randomly selected pairs of individuals interact if their opinions are sufficiently close, a natural approach is to rewire the social network according to these same distances \cite{kan2023adaptive,yu2017opinion,gargiulo2009opinion,de2022modelling}. Random network rewiring has also been used in models with binary opinions, see for example \cite{borges2024social,min2023coevolutionary} and the review \cite{vazquez2013opinion}, as well as in a version of the HK model \cite{su2014coevolution}. 

It is often typical to rewire according to the principle of homophily, meaning that individuals with close opinions, who are therefore deemed to be more similar, are more likely to be connected. Other mechanisms include triadic closure \cite{asikainen2020cumulative} and network updates based on structural balance theory \cite{kang2022coevolution}. In \cite{djurdjevac2024co}, individuals instead have a position in a `social space' alongside their opinions, which determines their network connections. Alternatively, in \cite{evans2018opinion}, both individuals' opinion and their social connections are modified as part of an evolutionary game, attempting to optimise a function combining the strength of their opinions, alignment with their neighbours and their total number of connections. 

Adaptive networks have also been applied to models in continuous time and opinion space, often related to the continuous time version of the HK model. In this setup, opinions evolve according to a set of ordinary or stochastic differential equations, and the network may evolve according to random rewiring \cite{iniguez2009opinion,liu2023emergence} or through its own set of differential equations \cite{nugent2023evolving,NugentGomesWolfram2024}. 

We focus here on this last case, specifically the model introduced in \cite{nugent2023evolving}, aiming to gain a deeper understanding of both long-term and short-time dynamics. The paper is organised as follows: the precise model setup and assumptions are introduced in Section \ref{Section: Model setup}. Section \ref{Section: Full support} studies the case where the function determining interactions is always strictly positive, while Section \ref{Section: Compact support} studies the bounded confidence case. Both sections include analytical results on the convergence to clusters and communities, extending previous results for a fixed network to cover various network dynamics. Section \ref{Section: Compact support} also includes numerical simulations where analytical results are not possible due to strong dependence on the initial network. The work in these sections provides a more complete understanding of the long-time dynamics of opinion formation under an adaptive network. Finally in Section \ref{Section: Early dynamics}, we examine the early weight dynamics before concluding in Section \ref{Section: Conclusion}. 

\section{Model Setup} \label{Section: Model setup}

This paper will primarily study the model introduced in \cite{nugent2023evolving}, which takes the form of a coupled system of ordinary differential equations (ODEs) for individuals' opinions and network edge weights. We first introduce the opinion dynamics; these are inspired by the classical HK model, its formulation in \cite{motsch_heterophilious_2014}, as well as the derivation from the discrete-interaction DW model discussed in \cite{nugent2024bridging}. 

Denote by $x_i\in[-1,1]$ the opinion of individual $i$, $w_{ij}\in[0,1]$ the edge weight between individuals $i$ and $j$, $k_i$ the degree of individual $i$ given by 
\begin{align*}
    k_i = \sum_{j=1}^N w_{ij},
\end{align*}
and finally $\phi:[-2,2]\rightarrow[0,1]$ the interaction function. The interaction function describes the likelihood or strength of interactions between two individuals as a function of the difference between their opinions. Opinions then evolve according to the following system of ODEs
\begin{align} \label{Eqn: Opinion ODEs}
    \frac{dx_i}{dt} &= \frac{1}{k_i} \sum_{j=1}^N w_{ij}\,\phi\big(x_j - x_i\big)\,(x_j - x_i).
\end{align}
In line with \cite{nugent2023evolving}, we introduce several assumptions to ensure this system is well-posed and fits within the framework of opinion formation models. 

\begin{assumption} \label{Assumption group: standard}
    The following are assumed to hold throughout. 
    \begin{enumerate}[label=\alph*)]
        \item The interaction function $\phi$ is Lipschitz continuous, with Lipschitz constant denoted $L_\phi$, and satisfies $\phi(r) = \phi(-r)$ for all $r\in[-2,2]$ and $\phi(0)>0$. \label{Assumptions on phi}
        \item Initial opinions are ordered, that is $x_1(0)\leq x_2(0) \leq\dots\leq x_N(0)$. \label{Assumptions on x0}
        \item All individuals give their own opinion weight $1$, that is $w_{ii}=1$ for all $i=1,\dots,N$. \label{Assumptions on w}
        \item The initial network $w(0)$ is strongly connected. That is, for any individuals $i,j$ there exists a sequence of individuals $i_n, n=0,\dots,m$ with $i_0=i, \, i_m=j$ and $w_{i_n,i_{n+1}}(0)>0$.
    \end{enumerate} 
\end{assumption}

The assumption on the continuity of $\phi$ is necessary for uniqueness of solutions to \eqref{Eqn: Opinion ODEs}, avoiding the issues discussed in \cite{blondel_continuous-time_2009} arising from the discontinuous bounded confidence interaction function. The assumption that $w_{ii}=1$ ensures that $k_i\geq1>0$ for all individuals, avoiding any possible issues in dividing by $k_i$ if individual $i$ is (or later, becomes) disconnected from the rest of the population. The assumption on the ordering of initial opinions is made only for simplicity and is not a necessary component of the model. The final assumption ensures that all parts of the network communicate with each other and thus may be considered as a single system, rather than being separated into distinct connected components. 

This model on a fixed, symmetric network has been studied in \cite{nugent2023evolving} where, using an energy argument similar to that in \cite{motsch_heterophilious_2014}, it is shown that the population converges to a clustered state in which individuals are partitioned into groups who share the same opinion but do not interact with each other. One such clustered state is consensus, in which all individuals hold the same opinion. This can be defined in several equivalent ways: firstly that there exists a value $x\in[-1,1]$ such that $x_i\rightarrow x$ for all $i=1,\dots,N$; secondly that the opinion diameter
\begin{align*}
    D(t) = \max_{i,j=1,\dots,N} |x_j(t) - x_i(t) | \,,
\end{align*}
converges to zero as $t\rightarrow\infty$. Another means to measure clustering is by defining an order parameter, originally introduced for opinion dynamics in \cite{wang2017noisy}, given by 
\begin{align} \label{eqn: order parameter}
    Q = \frac{1}{N^2 M_\phi} \sum_{i,j=1}^N \phi(x_j - x_i)  \,,
\end{align}
where $M_\phi$ is the maximum value of $\phi(r)$ over $r\in[-2,2]$. Assuming $\phi$ takes its maximum at $\phi(0)=1$, as with exponential or bounded confidence interaction functions, then $Q^{-1}$ is approximately the number of clusters with $Q=1$, giving a third equivalent description of consensus. This order parameter will be used in Section \ref{Section: Simulations BC} to measure the outcome of simulations. 

We now introduce the full coupled opinion and network dynamics proposed in \cite{nugent2023evolving}. The system is given by the following ODEs:
\begin{subequations} \label{Eqn: general dynamic network system}
\begin{align}
    \frac{dx_i}{dt} &= \frac{1}{k_i(t)} \sum_{j\neq i} w_{ij}\,\phi\big(x_j - x_i\big)\,(x_j - x_i) \,, & i = 1,\dots,N \label{Eqn: general dynamic network system x_i}\\
    \frac{dw_{ij}}{dt} &= \phi\big(x_j - x_i\big)\, f^+(w)_{ij} - \Big(1 - \phi\big(x_j - x_i\big)\Big)\, f^-(w)_{ij} \,, & i,j = 1,\dots,N,\, i\neq j \label{Eqn: general dynamic network system w_ij} \\[0.5em]
    \frac{dw_{ii}}{dt} &= 0, & i = 1,\dots,N.
\end{align}
\end{subequations}
The motivation for this structure is to use the interaction function to balance an edge growth term ($f^+(w)_{ij}$) against an edge removal term ($f^-(w)_{ij}$). These are functions of the entire network to allow for network dynamics that depend on other edges, for example triadic closure. 

Before introducing several assumptions on the form of $f^+$ and $f^-$, we define the following sets
\begin{align*}
    \mathcal{W}_{ij}^0 &= \{ w \in [0,1]^{N\times N} : w_{ij} = 0\}, \\
    \mathcal{W}_{ij}^1 &= \{ w \in [0,1]^{N\times N} : w_{ij} = 1\} \,.
\end{align*}

\begin{assumption} \label{Assumption group: weight dynamics}
    We make the following assumptions on the network dynamics:
    \begin{enumerate}[label=\alph*)]
        \item The functions $f^+:[0,1]^{N \times N} \rightarrow \mathbb{R}_{\geq0}^{N \times N}$ and $f^-:[0,1]^{N \times N} \rightarrow \mathbb{R}_{\geq0}^{N \times N}$ are both non-negative and Lipschitz continuous. 
        \item For all $w \in \mathcal{W}_{ij}^1$, $f^+(w)_{ij} = 0$. 
        \item For all $w \in \mathcal{W}_{ij}^0$, $f^-(w)_{ij} = 0$. 
    \end{enumerate}
\end{assumption}

These assumptions ensure that the system \eqref{Eqn: general dynamic network system} is well-posed and all edge weights remain in the interval $[0,1]$. In this paper we will focus on two edge weight dynamics introduced in \cite{nugent2023evolving}, namely memory weight dynamics and logistic weight dynamics. 

Memory weight dynamics simply bring the edge weight towards the interaction function, providing a `memory' of individuals' interaction history. By taking $f^+(w)_{ij} = (1-w_{ij})$ and $f^-(w)_{ij} = w_{ij}$ we have
\begin{align}
    \label{eqn: memory weight dynamics}
    \frac{dw_{ij}}{dt}
    &= \phi(x_j - x_i) - w_{ij} \,.
\end{align}
This choice of weight dynamics can both create and remove edges, and thus replaces much of the initial network structure and can have a major impact on the opinion formation process \cite{nugent2023evolving}. Naturally, the network clustering behaviour can be expected to mirror the clustering behaviour that is typical in opinion dynamics. 

Logistic weight dynamics change edge weights only, without creating any new edges. We take $f^+(w)_{ij} = f^-(w)_{ij} = w_{ij}(1 - w_{ij})$ giving
\begin{align}
    \label{eqn: logistic weight dynamics}
    \frac{dw_{ij}}{dt} 
    &= w_{ij}\,(1-w_{ij})\,(2\phi(x_j-x_i) - 1 ) \,.
\end{align}
Here, the initial network structure plays a much greater role as edges cannot be created, and so we may expect less predictability in both opinion and network dynamics. 

The following sections describing the long-term dynamics are split to address the different behaviours arising from two types of interaction functions. To separate these categories of interaction functions, we first introduce the following definition from \cite{NugentGomesWolfram2024}.
\begin{definition}
     For a given interaction function $\phi:[-2,2]\rightarrow[0,1]$ we denote the set of roots of $\phi$ by 
    \begin{equation}
        \mathcal{R}_\phi = \{ r \in [0,2] : \phi(|r|) = 0 \}.
    \end{equation}
    If $\mathcal{R}_\phi$ is empty then define $R = 2$, otherwise let $R = \inf (\mathcal{R}_\phi)$. 
\end{definition}
In Section \ref{Section: Full support} we consider the behaviour of the coupled opinion and network dynamics \eqref{Eqn: general dynamic network system} for interaction functions with $R=2$ and a variety of weight dynamics. In Section \ref{Section: Compact support}, we then consider interaction functions with $R<2$. 

\section{Interaction functions with full support} \label{Section: Full support}

As $R=2$, we can deduce that $\phi(r)>0$ for all $|r|<2$. If it is the case that $\phi(r)>c>0$ for all $r \in [-2,2]$ then individuals will always interact to some extent. As a result, the only way to avoid consensus is for the weight dynamics to separate the population. This is the main scenario we wish to consider in this section, hence we first address the specific case that $\phi(2)=0$. 

As all opinions remain inside the range of the initial opinions \cite{nugent2023evolving}, the value of $\phi(\pm2)$ is only relevant if at time $t=0$ there exists at least one individual with opinion $+1$ and at least one individual with opin ion $-1$ (otherwise there will be no distances of length 2 and so the value of $\phi(2)$ is irrelevant). As initial opinions are ordered we will therefore have that $x_1(0)=-1$ and $x_N(0)=1$. If all individuals have opinions in the set $\{-1,1\}$, then the opinion dynamics will be at a fixed point, and all weights will converge according to 
\begin{align*}
    \frac{dw_{ij}}{dt} =
    \begin{cases}
        \phi(0)\, f^+(w)_{ij} - \big(1 - \phi(0)\big)\, f^-(w)_{ij}, & \text{ if } x_i(0)=x_j(0), \\[0.4em]
        - f^-(w)_{ij}, & \text{ if } x_i(0) \neq x_j(0)
    \end{cases}\,.
\end{align*}

Otherwise, there must be some individual $i$ with an opinion in the open interval $(-1,1)$. Since the initial network is connected, there exists a path in $w(0)$ connecting individual $i$ to both individual $1$ and individual $N$. As the weight dynamics have finite speed there exists an $\varepsilon>0$ and a time $T>0$ such that all edges in this same path have weight at least $\varepsilon$ for all $0\leq t \leq T$. Considering the individuals neighbouring individual $i$ in this path, their opinions must immediately enter the open interval $(-1,1)$, if they don't begin there, as they interact with $x_i$. Once this occurs the same must happen for the individuals connected to them and so on until individuals $1$ and $N$ enter this open interval. At this point, everyone will have opinions in this open interval. The value of $\phi(\pm2)$ is therefore irrelevant after an arbitrarily small amount of time. 

As the focus of this section is on the emergent network structure and long-time behaviour of \eqref{Eqn: general dynamic network system}, we conclude that the only case in which $\phi(2)=0$ must be considered separately is when $x_i\in\{-1,1\}$ for all $i$, which is addressed above. For the remainder of this section we include the following additional assumption on $\phi$.  

\begin{assumption} \label{Assumption: phi > c}
    There exists a constant $c_\phi>0$ such that $\phi(r)>c_\phi>0$ for all $r\in[-2,2]$. 
\end{assumption}

In the case of a fixed network, Assumption \ref{Assumption: phi > c} is sufficient to guarantee convergence to consensus, see for example Proposition 2.2 in \cite{nugent2023evolving}. We now consider several examples of weight dynamics introduced above and address whether a similar result holds. 

\subsection{Memory weight dynamics}

In the case of memory weight dynamics it is possible to establish a similar result for the long-time behaviour of \eqref{Eqn: general dynamic network system}. Proposition \ref{Prop: edge creating dynamics} below applies to memory weight dynamics but we prove here a slightly more general version. 

\begin{proposition} \label{Prop: edge creating dynamics}
    Make Assumption \ref{Assumption: phi > c} and assume also that there exists a constant $c_f$ such that $f^+(w)_{ij} > c_f$ for all $w \in \mathcal{W}_{ij}^0$. Then the population reaches consensus. 
\end{proposition}
\begin{proof}
    As $f^+$ is Lipschitz continuous and by assumption $w\in \mathcal{W}_{ij}^0$ implies that $f^+(w)_{ij}>c_f$, there exists a constant $w^+>0$ such that $f^+(w)_{ij}>\frac{1}{2}c_f$ for all $w_{ij}<w^+$. 
    As $f^-$ is also Lipschitz continuous and $w\in \mathcal{W}_{ij}^0$ implies that $f^-(w)_{ij}=0$, there exists a constant $w^->0$ such that for all $w_{ij}<w^-$
    \begin{align*}
        f^-(w)_{ij} < \frac{c_\phi}{2(1 - c_\phi)} \, c_f.
    \end{align*} 
    Note that these fractions of $c_f$ are chosen to simplify the inequalities below. 
    
    Let $w^* = \min\{w^+,w^-\}$. Then for $w_{ij} < w^*$ we have 
    \begin{align*}
        \frac{dw_{ij}}{dt}
        &= \phi\big(x_j - x_i\big)\, f^+(w)_{ij} - \Big(1 - \phi\big(x_j - x_i\big)\Big)\, f^-(w)_{ij} \\
        &\geq c_\phi\, f^+(w)_{ij} - (1-c_\phi) f^-(w)_{ij} \\
        &\geq \frac{c_\phi c_f}{2} - (1-c_\phi) f^-(w)_{ij} \\
        &> \frac{c_\phi}{2} c_f - (1-c_\phi) \frac{c_\phi}{2(1 - c_\phi)} \, c_f 
        = 0 \,.
    \end{align*}
    Therefore $w_{ij}$ has positive derivative whenever $w_{ij}<w^*$ and so weights beginning above $w^*$ will remain above it and those beginning below will increase towards it, so there exists a finite time $t^*$ such that $w_{ij}(t)>\frac{1}{2}w^*$ for all $t\geq t^*$. 

    For any $t>t^*$, consider the individual with the minimum opinion, denoted $x_m$. The change in their opinion is given by
    \begin{align*}
        \frac{dx_m}{dt} 
        &= \frac{1}{k_m} \sum_{j=1}^N w_{mj}\,\phi(x_j - x_m)\,(x_j - x_m) 
        \geq \frac{c_\phi w^*}{2}\frac{1}{k_m} \sum_{j=1}^N (x_j - x_m) 
        \geq \frac{c_\phi w^*}{2} D(t).
    \end{align*}
    Similarly 
    \begin{align*}
        \frac{dx_M}{dt} 
        &\leq -\frac{c_\phi w^*}{2} D(t) \,.
    \end{align*}
    Hence 
    \begin{align*}
        \frac{d}{dt}D(t) \leq -c_\phi w^* D(t) \,,
    \end{align*}
    thus $D(t)\rightarrow0$ and the population reaches consensus. 
\end{proof}

We now apply this to memory weight dynamics and extend the result to describe the long-term network behaviour. 

\begin{corollary} \label{Cor: Exp, memory weight dynamics}
    Under memory weight dynamics and Assumption \ref{Assumption: phi > c}, the population reaches consensus, and the network becomes fully connected with edge weight $\phi(0)>0$.  
\end{corollary}
\begin{proof}
    Firstly note that for $w \in \mathcal{W}_{ij}^0$, $f^+(w)_{ij} = 1$, so we may immediately apply Proposition \ref{Prop: edge creating dynamics} to guarantee consensus. Thus we will have that $\phi(x_j - x_i) \rightarrow \phi(0)$ for all $i,j$ and so the solution to \eqref{eqn: memory weight dynamics} follows
    \begin{align*}
        w_{ij}(t) = e^{-t} w_{ij}(0) + \int_0^t e^{s-t} \phi\big( x_j(s) - x_i(s) \big) \, ds \rightarrow \phi(0) \,.
    \end{align*}
\end{proof}

\begin{remark}
    Under the assumptions of Proposition \ref{Prop: edge creating dynamics} it follows that $\phi(x_j - x_i)\rightarrow \phi(0)$. However, regarding other admissible weight dynamics, it is not possible to deduce anything about the long-term behaviour of $w$ in general (except that $w_{ij}$ will eventually exceed $w^*$) as there may be other fixed points in the weight dynamics. 
\end{remark}

\begin{proposition}
    Under memory weight dynamics and Assumption \ref{Assumption: phi > c}, the following bound holds on the diameter $D(t)$, 
    \begin{align*}
        D(t) \leq \exp\Big( - c^2 (t + e^{-t} - 1) \Big) D(0) \,.
    \end{align*}
\end{proposition}
\begin{proof}
    This is a direct application of the results in Section 2.1 of \cite{motsch_heterophilious_2014}, using the fact that 
    \begin{align*}
        w_{ij}(t) = e^{-t} w_{ij}(0) + \int_0^t e^{s-t} \phi\big( x_j(s) - x_i(s) \big) \, ds \geq e^{-t} w_{ij}(0) + c(1 -e^{-t}) \geq c(1 -e^{-t})\,,
    \end{align*}
    to give a lower bound on the strength of interaction terms. 
\end{proof}
This is very similar to the bound obtainable directly from the proof of Corollary \ref{Cor: Exp, memory weight dynamics} as for memory weight dynamics $c_\phi = w^* = c$. In both cases convergence may be significantly faster than this bound as it accounts only for the smallest possible value of $\phi$. 

In the case of $R=2$ and memory weight dynamics, there is a consistently high level of interaction between all individuals, making convergence to consensus and subsequent network converge quite predictable. We therefore move to consider more general weight dynamcis that allow more complex behaviours. 

\subsection{Logistic weight dynamics}

This section begins with a general result for interaction functions with full support, provides several results specific to logistic weight dynamics, and finally discusses a scenario that highlights a barrier to proving results about the positions of opinion clusters.  

\begin{figure}[ht!]
    \centering
    \begin{subfigure}{\textwidth}
        \centering
        \includegraphics[width=0.75\linewidth]{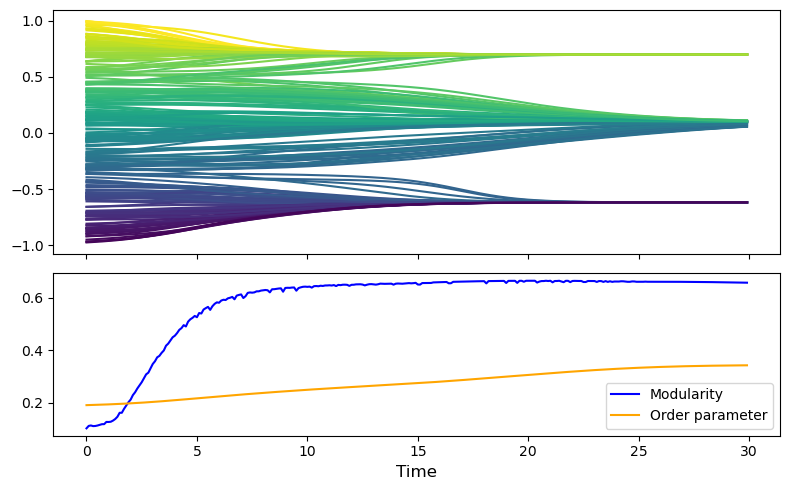}
        \caption{The top panel shows opinion trajectories over time, coloured according to initial opinion. Three distinct clusters form as weight dynamics removes their connecting edges. The bottom panel shows the order parameter and network modularity over time. The modularity rises sharply at the start of the dynamics as the network separates into communities.}
    \end{subfigure}
    \begin{subfigure}{\textwidth}
        \centering
        \includegraphics[width=0.75\linewidth]{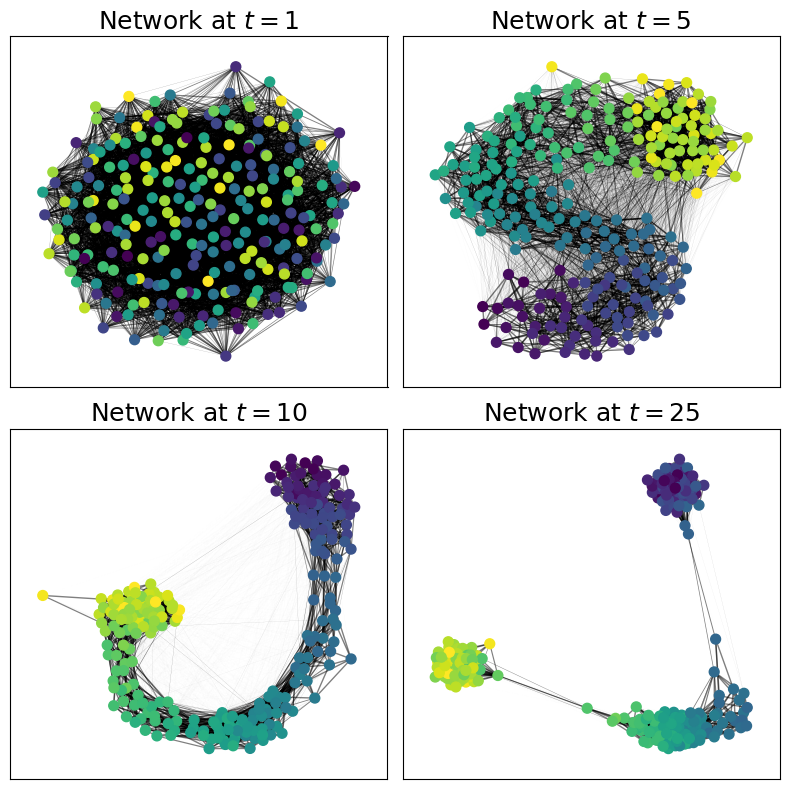}
        \caption{Snapshots of network structure. Edge width is drawn according to edge weight; nodes are coloured according to individuals' initial opinions. The network first becomes ordered according to individuals' opinions before separating into three communities. Note that nodes are positioned according to the network structure, rather than an embedding into some physical space.}
    \end{subfigure}
    \caption{Demonstration of opinion clustering under logistic weight dynamics. In the absence of weight dynamics the interaction function having full support would lead to consensus. Here, the network is instead reorganised according to individuals' opinions, leading to the formation of communities that are then mirrored in the formation of opinion clusters.}
    \label{fig:exp_logistic_clusters}
\end{figure}

Firstly we clarify that unlike in the case of memory weight dynamics, consensus is not guaranteed under logistic weight dynamics. If $\phi(x_j-x_i)<1/2$ the edge weight $w_{ij}$ will decrease, potentially separating the (initially connected) network into disconnected components. An example of this behaviour is shown in Figure \ref{fig:exp_logistic_clusters} where $\phi$ is an exponential interaction function, initial opinions are chosen uniformly at random and the initial network is an Erdos-Renyi random network with $p=0.25$. The top panel shows opinion trajectories over time, coloured according to initial opinions, which demonstrates a split into three opinion clusters. While the order parameter \eqref{eqn: order parameter} increases gradually over time there is an earlier, sharper increase in the network modularity. Modularity describes the extent of community structure in a network by calculating the difference between the number of within-group edges from the expected number in an equivalent network with edges placed at random \cite{newman2006modularity}. The increase in modularity indicates the emergence of community structure in the network and occurs alongside the formation of opinion clusters. 

This structural change can be seen in the network snapshots in the bottom part of Figure \ref{fig:exp_logistic_clusters}. The initially random network organises, effectively ordering itself according to individuals initial opinions. This behaviour is different to the immediate formation of clusters often observed under rewiring and is more comparable to the behaviour observed in \cite{djurdjevac2024co} in which individuals are positioned in a `social space'. This change in network structure creates a positive feedback loop with opinion clustering that, over time, leads to the formation of distinct opinion clusters and corresponding network communities. Note that as logistic weight dynamics does not create new edges these communities are densely, but not fully, connected. This setup, in particular the transition from consensus to clustering, is studied through simulations in Section 3.3 of \cite{nugent2023evolving}. Below we provide an analytic result confirming the observed cluster formation for a broad class of weight dynamics. 

\begin{theorem} \label{Thm: Exp, logistic weight dynamics}
    Under Assumption \ref{Assumption: phi > c} and any weight dynamics satisfying Assumption \ref{Assumption group: weight dynamics}, the population separates into finitely many opinions clusters with all edges between clusters vanishing. That is, for any two individuals $i,j$, at least one of $x_j - x_i$ or $w_{ij}$ converges to zero as $t\rightarrow \infty$. 
\end{theorem}
\begin{proof}    
    Consider the individual with the minimum opinion $x_m(t)$. This value is non-decreasing and bounded above by $x_N(0)$, so must converge to some value as $t\rightarrow\infty$. Denote this value $x_m^*$. Define $g(t)$ by
    \begin{align*}
        \frac{dx_m}{dt} = \frac{1}{k_m} \sum_{j=1}^N w_{mj}\,\phi(x_j - x_m)\,(x_j - x_m) =: g(t)\,.
    \end{align*}
    Note that $g$ is positive and is not necessarily continuous in time. Specifically, if individual $i$ currently holds the minimum opinion and $x_i$ is increasing quickly it may cross the opinion of another individual $j$, causing a discontinuous drop in $g(t)$ as the identity of the individual with the minimum opinion changes. Such discontinuities must always decrease $g(t)$, as the minimum could never catch up with an opinion that was increasing faster than it. Moreover, in the periods between these discontinuities $g(t)$ is Lipschitz continuous as all of $x,w,k$ and $\phi$ are also Lipschitz continuous (as a consequence of Assumption \ref{Assumption group: standard}).

    We claim that $g(t) \rightarrow0$. Assuming otherwise, then as $g(t)\geq0$ it must take values above some $\varepsilon>0$ infinitely often. As $g(t)$ is continuous when increasing, it must at some earlier time pass through $\varepsilon/2$. Moreover, as it is Lipschitz when increasing this time difference must be above some fixed $s>0$. Hence each excursion above $\varepsilon$ contributes at least $s \varepsilon/2$ towards the value of $x_m$ and so $x_m \rightarrow \infty$. This is a contradiction as $x_m$ is bounded above by $x_N(0)$, hence
    \begin{align*}
        \frac{1}{k_m} \sum_{j=1}^N w_{mj}\,\phi(x_j - x_m)\,(x_j - x_m) \rightarrow 0 \,.
    \end{align*}
    As all terms in this sum are positive (and $k_m$ is bounded above by $N$) we must have that for all $j$, 
    \begin{align*}
        w_{mj}\,\phi(x_j - x_m)\,(x_j - x_m) \rightarrow 0 \,.
    \end{align*}
    As $\phi(x_j - x_i) > c$, either $x_j \rightarrow x_m^*$  or $x_j \nrightarrow x_m^*$ and so $w_{mj} \rightarrow 0$. Denote by $(G1)$ the group of individuals for which $x_j \rightarrow x_m^*$ and by $(G2)$ those for whom $x_j \nrightarrow x_m^*$ and so $w_{mj} \rightarrow 0$. These groups partition the population. 
    
    As the minimum opinion converges to $x_m^*$ there must be at least one individual in $(G1)$, but $(G2)$ may be empty. If $(G2)$ is indeed empty then we have consensus, meaning $x_j-x_i\rightarrow0$ for all $i,j$, so we are done. 
    
    Hence assume $(G2)$ is non-empty. Let $i\in(G1)$ and consider
    \begin{align*}
        \frac{dx_i}{dt} = \frac{1}{k_i} \bigg( \sum_{j\in(G1)} w_{ij}\,\phi(x_j - x_i)\,(x_j - x_i) + \sum_{j\in(G2)} w_{ij}\,\phi(x_j - x_i)\,(x_j - x_i) \bigg) =: g_1(t) + g_2(t)\,.
    \end{align*}
    It follows immediately that $g_1(t) \rightarrow 0$ as $x_j-x_i\rightarrow0$ for all $j\in(G1)$. As $i\in (G1)$, $x_i$ has a limit as $t\rightarrow\infty$, which implies that 
    \begin{align*}
        \int_0^\infty \frac{dx_i(s)}{dt} \, ds < \infty \,.
    \end{align*}
    As we know $g_1(t)\rightarrow0$, it follows that 
    \begin{align*}
        \int_0^\infty g_2(t) \, dt < \infty \,.
    \end{align*}
    Moreover $g_2(t)$ is uniformly continuous, hence by Barablat's lemma \cite{barbalat1959systemes} $g_2(t)\rightarrow0$. Thus for all $j\in(G2)$, 
    \begin{align*}
        w_{ij}\,\phi(x_j - x_i)\,(x_j - x_i) \rightarrow 0 \,.
    \end{align*}
    As $\phi > c$ and $x_j \nrightarrow x_m^*$ the only remaining possibility is that $w_{ij} \rightarrow 0$. To summarise, we have argued here that for all $i\in (G1)$ and $j \in (G2)$, $w_{ij} \rightarrow 0$. 
        
    Considering next the minimum element of $(G2)$, denoted $x_{\Tilde{m}}$ we again separate
    \begin{align*}
        \frac{dx_{\Tilde{m}}}{dt} = \frac{1}{k_{\Tilde{m}}} \bigg( \sum_{j\in(G1)} w_{\Tilde{m},j}\,\phi(x_j - x_{\Tilde{m}})\,(x_j - x_{\Tilde{m}}) + \sum_{j\in(G2)} w_{\Tilde{m},j}\,\phi(x_j - x_{\Tilde{m}})\,(x_j - x_{\Tilde{m}}) \bigg) =: g_1(t) + g_2(t)\,.
    \end{align*}
    As $w_{ij}\rightarrow0$ for all $i\in(G1)$ and $j\in(G2)$, $g_1(t)\rightarrow0$ and using similar arguments to those above for $g$ we must also have $g_2(t)\rightarrow0$. Denote the limiting value of $x_{\Tilde{m}}$ by $x_{\Tilde{m}}^*$. Again we may separate $(G2)$ into two populations $(G21)$ and $(G22)$ with individuals in the former converging to $x_{\Tilde{m}}^*$ and those in the latter not doing so. As before this implies $w_{ij}\rightarrow0$ for all $i\in(G21)$ and $j\in(G22)$. 

    We then apply the whole argument to $(G22)$ and repeat. As there are a finite number of individuals we will eventually have that some $(G\dots2)$ group will be empty and all individuals in the population have been covered. All individuals have thus been separated using this procedure into opinion clusters. Moreover, we have shown that all edges between clusters must vanish as required. 
\end{proof}

Note that Theorem \ref{Thm: Exp, logistic weight dynamics} also implies that $x_i(t)$ has a limit as $t\rightarrow\infty$ for all $i$. This clustering result extends similar results for a fixed network to the dynamic network case. 

We next apply this specifically to logistic weight dynamics. 

\begin{corollary}
    Under Assumption \ref{Assumption: phi > c} and logistic weight dynamics, if $\phi>c\geq 1/2$ then the population reaches consensus and $w_{ij} \rightarrow \mathds{1}\{w_{ij}(0)>0\}$. 
\end{corollary}
\begin{proof}
    If $\phi > c$ with $c \geq 1/2$, then all non-zero weights will have positive derivative. Thus if $w_{ij}(0)>0$ it is not possible for $w_{ij}\rightarrow0$ and so by Theorem \ref{Thm: Exp, logistic weight dynamics} $x_j - x_i \rightarrow 0$ for all such pairs. As the initial network is assumed to be strongly connected we will therefore have consensus. If $w_{ij}=0$ then this is fixed by the dynamics. 
\end{proof}

\begin{corollary} \label{Corollary: limit of w_ij (exp,logistic)}
    Under Assumption \ref{Assumption: phi > c} and logistic weight dynamics, for all pairs $i,j$, $w_{ij}\rightarrow1$ if all the following conditions are met:
    \begin{enumerate}[label=\alph*)]
        \item $x_j - x_i \rightarrow 0$,
        \item $w_{ij}(0) > 0$,
        \item $2\phi(0)>1$,
    \end{enumerate}
    otherwise $w_{ij}\rightarrow0$. 
\end{corollary}
\begin{proof}
    If $x_j - x_i \rightarrow 0$ then $\phi(x_j - x_i) \rightarrow \phi(0)$ and so $w_{ij}\rightarrow 1$ if and only if $w_{ij}(0) > 0$ and $2\phi(0)-1>0$ (as $0$ is a fixed point of the logistic weight dynamics). If $x_j - x_i \nrightarrow 0$ then by Theorem \ref{Thm: Exp, logistic weight dynamics} we must have that $w_{ij}\rightarrow0$. 
\end{proof}

\begin{corollary}
    Under Assumption \ref{Assumption: phi > c} and logistic weight dynamics, for a pair of individuals $i$ and $j$, if $w_{ij}(0)=1$ then individuals $i$ and $j$ will lie in the same opinion cluster regardless of the rest of the population's behaviour.
\end{corollary}
\begin{proof}
    As $w_{ij}=1$ is a fixed point of the logistic weight dynamics it is not possible for $w_{ij}\rightarrow 0$. Thus by Theorem \ref{Thm: Exp, logistic weight dynamics} we must have that $x_j - x_i \rightarrow 0$, meaning they will lie in the same opinion cluster. 
\end{proof}

\begin{remark} \label{Remark: Not possible to identify cluster locations}
    Unlike a fixed network or the case of memory weight dynamics, for logistic weight dynamics we do not show convergence to consensus as this is not guaranteed unless $\phi>c\geq 1/2$. Moreover, for logistic weight dynamics it is not possible to argue in general about the positions of the opinion clusters. Consider the scenario shown in Figure \ref{fig:diagram}. The population has two large opinion clusters close to each other, with the individuals in these clusters denoted by $\mathcal{C}_1$ and $\mathcal{C}_2$ respectively. Assume that $w(0)$ is connected within each of these clusters. Assume also that there exists a single individual with a much higher opinion, denoted $x_M$. As the only requirement on the initial network structure is that it is connected, it is entirely possible that there are no initial edges, and thus no possible creation of edges, between $\mathcal{C}_1$ and $\mathcal{C}_2$, as they may be connected only via individual $M$. As $x_M$ is extremely far away these connections would decay quickly, causing the limiting network to be disconnected and the two nearby clusters to persist without merging. Note that this setup only requires that the interaction function be smaller than a half after some distance. 
    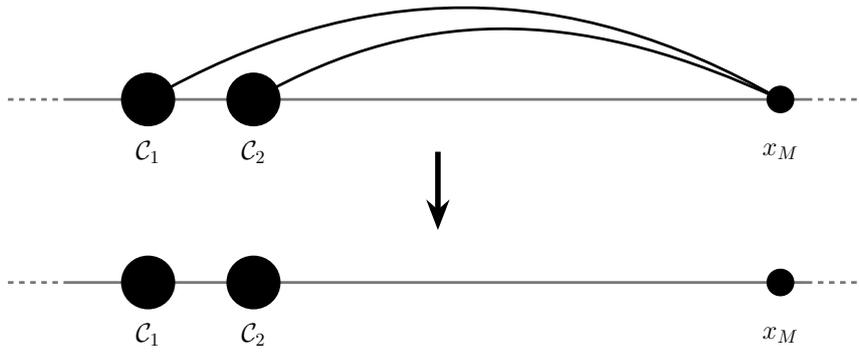
\begin{figure}[!ht]
\centering
\resizebox{.7\textwidth}{!}{%
\begin{circuitikz}
\tikzstyle{every node}=[font=\LARGE]

\draw [line width=1.5pt, color={rgb,255:red,120; green,120; blue,120}, short] (5.5,11) -- (19.5,11);
\draw [line width=1.5pt, color={rgb,255:red,120; green,120; blue,120}, dashed] (5.5,11) -- (4.25,11);
\draw [line width=1.5pt, color={rgb,255:red,120; green,120; blue,120}, dashed] (19.5,11) -- (20.75,11);


\draw [line width=1.5pt,color={rgb,255:red,0; green,0; blue,0}] (7,11) to[out=30,in=150] (19,11);
\draw [line width=1.5pt,color={rgb,255:red,0; green,0; blue,0}] (9,11) to[out=30,in=155] (19,11);

\draw [ fill={rgb,255:red,0; green,0; blue,0} , line width=0.8pt ] (7,11) circle (0.5cm);
\draw [ fill={rgb,255:red,0; green,0; blue,0} , line width=0.8pt ] (9,11) circle (0.5cm);
\draw [ fill={rgb,255:red,0; green,0; blue,0} , line width=0.8pt ] (19,11) circle (0.25cm);

\node [font=\Large] at (7,10) {$\mathcal{C}_1$};
\node [font=\Large] at (9,10) {$\mathcal{C}_2$};
\node [font=\Large] at (19,10) {$x_M$};


\draw [line width=3pt, ->, >=Stealth] (12.5,10) -- (12.5,8.5);


\draw [line width=1.5pt, color={rgb,255:red,120; green,120; blue,120}, short] (5.5,7.5) -- (19.5,7.5);
\draw [line width=1.5pt, color={rgb,255:red,120; green,120; blue,120}, dashed] (5.5,7.5) -- (4.25,7.5);
\draw [line width=1.5pt, color={rgb,255:red,120; green,120; blue,120}, dashed] (19.5,7.5) -- (20.75,7.5);

\draw [ fill={rgb,255:red,0; green,0; blue,0} , line width=0.8pt ] (7,7.5) circle (0.5cm);
\draw [ fill={rgb,255:red,0; green,0; blue,0} , line width=0.8pt ] (9,7.5) circle (0.5cm);
\draw [ fill={rgb,255:red,0; green,0; blue,0} , line width=0.8pt ] (19,7.5) circle (0.25cm);

\node [font=\Large] at (7,6.5) {$\mathcal{C}_1$};
\node [font=\Large] at (9,6.5) {$\mathcal{C}_2$};
\node [font=\Large] at (19,6.5) {$x_M$};

\end{circuitikz}
}%

\caption{Diagram showing the scenario described in Remark \ref{Remark: Not possible to identify cluster locations}. Under logistic weight dynamics two clusters ($\mathcal{C}_1$ and $\mathcal{C}_2$) may exist at any distance apart due to an initial network structure in which they are only connected via a distant individual ($x_M$) from whom they both become disconnected.}
\label{fig:diagram}
\end{figure}
\end{remark}

Theorem \ref{Thm: Exp, logistic weight dynamics} also holds true in the case of Friend-Of-A-Friend (FOAF) weight dynamics introduced in \cite{nugent2023evolving}, as these satisfy Assumption \ref{Assumption group: weight dynamics}. The weight dynamics are given by 
\begin{align*}
    \frac{dw_{ij}}{dt} =  &= \phi(x_j - x_i)\,\big( w_{ij} + N^{-1} (w^2)_{ij} \big)\,(1-w_{ij}) - \big(1 - \phi(x_j - x_i)\big)\,w_{ij}  \,,
\end{align*}
and capture the idea of triadic closure (the creation of a new edge between two individuals that have a mutual connection/friend) that has been used in various models of social network evolution \cite{bianconi2014triadic,klimek2013triadic,huang2018will,asikainen2020cumulative}. 

As in Corollary \ref{Corollary: limit of w_ij (exp,logistic)} we consider two cases. Again if $x_j - x_i \nrightarrow 0$ then we must automatically have that $w_{ij}\rightarrow0$. If $x_j - x_i \rightarrow 0$ then the situation is more complex as it possible to create weights through mutual connections. If $w_{ij}>0$ at any point then we will have that $w_{ij}\rightarrow1$ as $t\rightarrow\infty$ since $\phi(x_j - x_i) \rightarrow 1$. Therefore in this case $w_{ij}\rightarrow0$ if and only if $w_{ij}(t)=0$ for all $t\geq0$, a statement which depends upon the exact trajectory of the whole system. However in general, for a fairly well connected initial network, one would expect the FOAF mechanism to lead to fully connected communities in the network corresponding to opinion clusters. Note that again Remark \ref{Remark: Not possible to identify cluster locations} applies and the locations of clusters cannot be determined. In this case, as demonstrated in Figure \ref{fig:diagram 2}, now place two individuals at $x_M$ so that $\mathcal{C}_1$ and $\mathcal{C}_2$ do not share $x_M$ as a mutual friend. The connections between each cluster and their neighbour at $x_M$ will decay quickly, as the FOAF mechanic will not increase the edge weight, causing the network to be disconnected. Thus the clusters can take any position sufficiently far from $x_M$. 

\begin{figure}[!ht]
\centering
\resizebox{.7\textwidth}{!}{%
\begin{circuitikz}
\tikzstyle{every node}=[font=\LARGE]

\draw [line width=1.5pt, color={rgb,255:red,120; green,120; blue,120}, short] (5.5,11) -- (19.5,11);
\draw [line width=1.5pt, color={rgb,255:red,120; green,120; blue,120}, dashed] (5.5,11) -- (4.25,11);
\draw [line width=1.5pt, color={rgb,255:red,120; green,120; blue,120}, dashed] (19.5,11) -- (20.75,11);

\draw [line width=1.5pt,color={rgb,255:red,0; green,0; blue,0}] (7,11) to[out=30,in=150] (19,11.75);
\draw [line width=1.5pt,color={rgb,255:red,0; green,0; blue,0}] (9,11) to[out=30,in=155] (19,11);

\draw [ fill={rgb,255:red,0; green,0; blue,0} , line width=0.8pt ] (7,11) circle (0.5cm);
\draw [ fill={rgb,255:red,0; green,0; blue,0} , line width=0.8pt ] (9,11) circle (0.5cm);
\draw [ fill={rgb,255:red,0; green,0; blue,0} , line width=0.8pt ] (19,11) circle (0.25cm);
\draw [ fill={rgb,255:red,0; green,0; blue,0} , line width=0.8pt ] (19,11.75) circle (0.25cm);

\draw [line width=1.5pt, color={rgb,255:red,0; green,0; blue,0}, short] (19,11) -- (19,11.75);

\node [font=\Large] at (7,10) {$\mathcal{C}_1$};
\node [font=\Large] at (9,10) {$\mathcal{C}_2$};
\node [font=\Large] at (19,10) {$x_M$};



\draw [line width=3pt, ->, >=Stealth] (12.5,10) -- (12.5,8.5);


\draw [line width=1.5pt, color={rgb,255:red,120; green,120; blue,120}, short] (5.5,7.5) -- (19.5,7.5);
\draw [line width=1.5pt, color={rgb,255:red,120; green,120; blue,120}, dashed] (5.5,7.5) -- (4.25,7.5);
\draw [line width=1.5pt, color={rgb,255:red,120; green,120; blue,120}, dashed] (19.5,7.5) -- (20.75,7.5);

\draw [ fill={rgb,255:red,0; green,0; blue,0} , line width=0.8pt ] (7,7.5) circle (0.5cm);
\draw [ fill={rgb,255:red,0; green,0; blue,0} , line width=0.8pt ] (9,7.5) circle (0.5cm);
\draw [ fill={rgb,255:red,0; green,0; blue,0} , line width=0.8pt ] (19,7.5) circle (0.25cm);
\draw [ fill={rgb,255:red,0; green,0; blue,0} , line width=0.8pt ] (19,8.25) circle (0.25cm);

\draw [line width=1.5pt, color={rgb,255:red,0; green,0; blue,0}, short] (19,7.5) -- (19,8.25);

\node [font=\Large] at (7,6.5) {$\mathcal{C}_1$};
\node [font=\Large] at (9,6.5) {$\mathcal{C}_2$};
\node [font=\Large] at (19,6.5) {$x_M$};

\end{circuitikz}
}%

\caption{Diagram showing the scenario described in extending Remark \ref{Remark: Not possible to identify cluster locations} to FOAF weight dynamics. Under FOAF weight dynamics two clusters ($\mathcal{C}_1$ and $\mathcal{C}_2$) may exist at any distance apart due to an initial network structure in which they are only connected via two distant connected individuals (both at $x_M$) from whom the clusters become disconnected.}
\label{fig:diagram 2}
\end{figure}
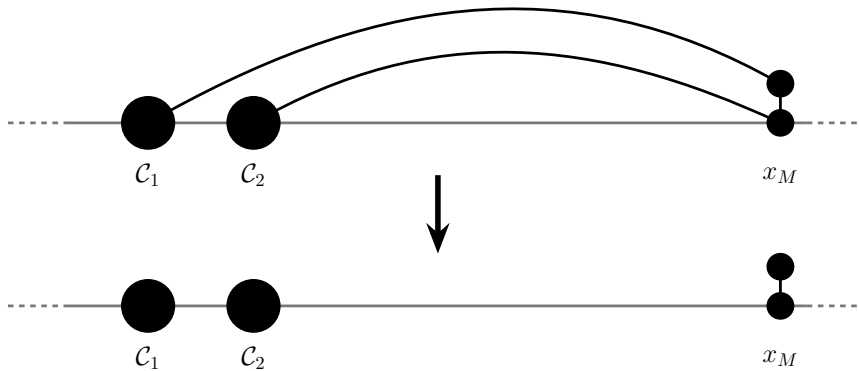

Lastly we may also apply Theorem \ref{Thm: Exp, logistic weight dynamics} to memory weight dynamics. In the proof of Corollary \ref{Cor: Exp, memory weight dynamics} it is shown that eventually $w_{ij}(t)$ is bounded below by some strictly positive constant. This implies that $w_{ij} \nrightarrow 0$ and so from Theorem \ref{Thm: Exp, logistic weight dynamics} we conclude that $x_j - x_i \rightarrow 0$ for all $i,j$. That is, the population reaches consensus. However this approach does not provide the exponential decrease in the diameter that is also shown by Corollary \ref{Cor: Exp, memory weight dynamics}. 

\section{Interaction functions with compact support}
\label{Section: Compact support}

In this section we consider interaction functions with $R<2$. Such interaction functions can be considered as generalisations of the idea of bounded confidence: that individuals are unwilling to interact with those whose opinions are too different from their own, except that this no longer takes the form of a discontinuous cut-off. We specifically consider interaction functions in the following class.

\begin{definition}
    An interaction function $\phi:[-2,2]\rightarrow\mathbb{R}$ is called a smooth bounded confidence interaction function with radius (or confidence bound) $R \in (0,2)$ if it is continuous, symmetric, $\phi(r)>0$ for $|r|<R$ and $\phi(r)=0$ for $|r|\geq R$. 
\end{definition}

Note that this does not require the interaction function to be monotonic, but does exclude the possibility of $\phi$ having additional roots strictly closer than $R$. 

In this case of a fixed network, such interaction functions give rise to the formation of (either singular or multiple) opinion clusters: collections of individuals who share the same opinion but lie a distance apart greater than $R$ \cite{nugent2023evolving}.  We next extend this result to the case of memory weight dynamics. 

\subsection{Memory weight dynamics}

Before providing the main result of this section, we first recall the definition of an $R$-chain from \cite{NugentGomesWolfram2024}.

\begin{definition}
    A group of individuals forms an $R$-chain if, when individuals are ordered by opinion, the distance between each pair of neighbours is strictly smaller than $R$. 
\end{definition}

Recall that, regardless of the network dynamics, when $\phi$ is a smoothed bounded confidence interaction function with radius $R$, if a gap of size at least $R$ is present in the population's opinions it will never close (\cite{nugent2023evolving} Proposition 3.2). Thus if individuals are in different $R$-chains then they will never interact. 

\begin{proposition} \label{Prop: BC, memory weight dynamics}
    Assume that $\phi$ is a smoothed bounded confidence interaction function with radius $R$. Then, for any two individuals $i,j$ either $x_j - x_i \rightarrow 0$ or $\phi(x_j - x_i)\rightarrow 0$ as $t\rightarrow \infty$, with $w_{ij}\rightarrow1$ or $w_{ij}\rightarrow0$ accordingly. 
\end{proposition}
\begin{proof}
    As distinct $R$-chains never interact we can consider them separately. In addition for $i,j$ in distinct $R$-chain $\phi(x_j - x_i)=0$ at all times, so the statement is immediately satisfied. Hence for simplicity we may assume that the entire population begins and remains an $R$-chain, as if it were to begin as, or break up into, multiple $R$-chains we could consider these independently from that point onwards. 
   
    Consider again the individual with the minimum opinion $x_m$. As in the proof of Theorem \ref{Thm: Exp, logistic weight dynamics} we must have that for all $j$, 
    \begin{align*}
        w_{mj}\,\phi(x_j - x_m)\,(x_j - x_m) \rightarrow 0 \,.
    \end{align*}
    All terms in the product are bounded and at least one must converge to zero. As we are considering memory weight dynamics, $w_{mj}\rightarrow0$ if and only if $\phi(x_j - x_m) \rightarrow 0$. We can therefore split the population into two groups: those for whom $x_j - x_m \rightarrow 0$ (meaning $x_j \rightarrow x_m^*$, the limiting minimum opinion) referred to as $(G1)$; and those for whom $\phi(x_j - x_m)\rightarrow 0$, referred to as $(G2)$. 

    The proposition statement is thus satisfied for individuals $i \in (G1)$. If $j\in(G1)$ also then both $x_i$ and $x_j$ converge to $x_m^*$ and so $x_j-x_i\rightarrow0$. If $j\in(G2)$ then $\lim_{t\rightarrow\infty}\phi(x_j - x_i) = \lim_{t\rightarrow\infty}\phi(x_j - x_m) = 0$ due to the continuity of $\phi$. 

    As in the proof of Theorem \ref{Thm: Exp, logistic weight dynamics} we now consider the minimum element of $(G2)$, with the goal of splitting this group and repeating a similar argument. 

    Note firstly that as $x_j \rightarrow x_m^*$ for all $j\in(G1)$, for any $\varepsilon>0$ there exists a time $t$ after which $|x_j - x_m^*|<\varepsilon$ for all $j\in(G1)$. Similarly, as $\phi(x_j-x_m)\rightarrow0$ for all $j \in (G2)$, for any $\varepsilon>0$ there exists a time $t$ after which $x_j > x_m^* + R - \varepsilon$ for all $j\in(G2)$. This is true since $R\in(0,2)$ is the smallest root of $\phi$. For a given $\varepsilon>0$ let $t^*$ be the maximum of these two times. 
    
    Let $\Tilde{m} \in (G2)$ be the individual in $(G2)$ with the minimum opinion. Using the above, $x_{\Tilde{m}} > x_m^* + R - \varepsilon$ after $t^*$. As the population must remain an $R$-chain, and after $t^*$ all individuals in $(G1)$ are within $\varepsilon$ of $x_m^*$, we conclude that $x_{\Tilde{m}} \in (x_m^* + R - \varepsilon, x_m^* + R + \varepsilon)$. Hence $x_{\Tilde{m}} \rightarrow x_m^* + R$. 
        
    We can split the derivative of $x_{\Tilde{m}}$ into two parts as follows
    \begin{align*}
        \frac{dx_{\Tilde{m}}}{dt} = \frac{1}{k_{\Tilde{m}}} \bigg( \sum_{j\in(G1)} w_{\Tilde{m},j}\,\phi(x_j - x_{\Tilde{m}})\,(x_j - x_{\Tilde{m}}) + \sum_{j\in(G2)} w_{\Tilde{m},j}\,\phi(x_j - x_{\Tilde{m}})\,(x_j - x_{\Tilde{m}}) \bigg) =: g_1(t) + g_2(t)\,.
    \end{align*}
    As $\Tilde{m}\in(G2)$ we have that $\phi(x_j - x_{\Tilde{m}})\rightarrow 0$ for all $j \in(G1)$ and so $g_1(t)\rightarrow 0$ as $t\rightarrow\infty$. In addition $g_2(t)$ is positive, as it contains only interactions within $(G2)$ and $x_{\Tilde{m}}$ is the minimum element of $(G2)$. As previously with $g$ in the proof of Theorem \ref{Thm: Exp, logistic weight dynamics}, $g_2$ may be discontinuous when decreasing but is uniformly Lipschitz continuous when increasing. Hence by the same argument, $g_2(t)\rightarrow 0$ as $t \rightarrow \infty$. 

    This brings us to the same position as before. Now for each $j\in (G2)$ either $x_j - x_{\Tilde{m}} \rightarrow 0$ (meaning $x_j \rightarrow x_m^* + R$) or $\phi(x_j - x_{\Tilde{m}})\rightarrow 0$. The argument above can then be repeated as many times as necessary until all individuals have been shown to converge to $x_m^* + nR$ for some $n\in\mathbb{N}$, which implies the result. 
\end{proof}

\begin{remark}
    In practice, for a bounded confidence interaction function, we do not expect the population to remain in a single $R$-chain. In fact the typical behaviour is that the population breaks up into clusters a distance apart further than $R$. However, as the result above concerns the long-time behaviour, it can simply be applied to each $R$-chain independently after this break-up has occurred. 
\end{remark}

\subsection{Logistic weight dynamics} \label{Section: BC, Logistic}

In this case it is not possible to describe the behaviour of the system without additional information about the initial network structure. To highlight this we again separate the population into groups as in the previous proofs. Some details are omitted as this construction is essentially the same as those above. 

We again choose an individual $i$ with $x_i \rightarrow x_m^*$. For all $j$, 
\begin{align*}
    w_{ij}\,\phi(x_j - x_i)\,(x_j - x_i) \rightarrow 0 \,.
\end{align*}
This now gives three options: $x_j \rightarrow x_m^*$ $(G1)$, $\phi(x_j-x_m^*)\rightarrow 0$ $(G2)$ or $w_{ij}\rightarrow 0$ for all such $i$ $(G3)$. As before later groups exclude earlier ones. 

Note that when considering $\phi$ with full support in Section \ref{Section: Full support} we could exclude $\phi(x_j-x_m^*)\rightarrow 0$, leaving only two options. In Proposition \ref{Prop: BC, memory weight dynamics} we could not exclude $\phi(x_j-x_m^*)\rightarrow 0$, but memory weight dynamics mean that $\phi(x_j-x_m^*)\rightarrow 0$ if any only if $w_{ij}\rightarrow 0$, again reducing to only two options. When considering a bounded confidence interaction function and logistic weight dynamics we cannot reduce the number of groups. 

The central problem is that nothing is known about the opinions of individuals in group $(G3)$. Recalling the construction in Remark \ref{Remark: Not possible to identify cluster locations}, in which an initially connected network becomes disconnected under logistic weight dynamics, it is also possible in this scenario for individuals in $(G3)$ to have essentially arbitrary opinions. As a result it is not possible to characterise in general the interactions between individuals in $(G2)$ and those in $(G3)$, meaning the argument cannot be applied repeatedly to gradually separate the population into opinion clusters.

Instead, we study this combination of bounded confidence interaction functions and logistic weight dynamics through simulations, extending the analysis performed in \cite{nugent2023evolving} for Erdos-Renyi (ER) random networks to different initial network structures, including Watts-Strogatz (WS) random networks, Barabasi-Albert (BA) random networks and a Stochastic Block Model (SBM). We first introduce these random networks and provide some intuition on their initial structure and how it develops, then describe the setup for more extensive simulations.

\begin{figure}[ht!]
    \centering
    \includegraphics[width=\linewidth]{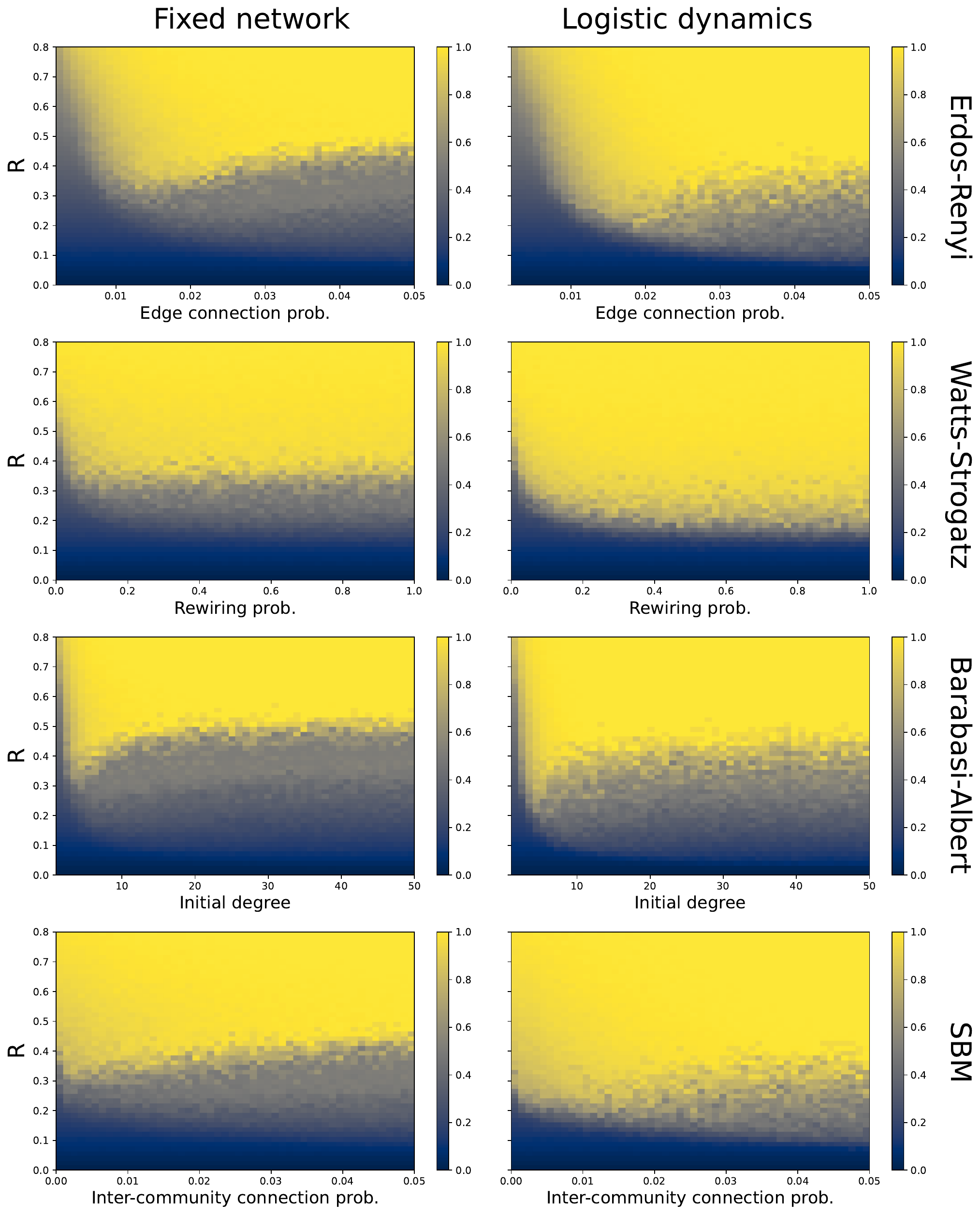}
    \caption{Heatmaps showing the order parameter \eqref{eqn: order parameter} at steady state for opinion formation with a fixed network (left) and logistic weight dynamics (right) using different random network models to generate $w(0)$. Blue areas show opinion clustering, grey areas indicate polarisation while yellow means the population has reached consensus. In each heatmap the cutoff value $R$ for a bounded confidence interaction function is varied on the vertical axis, while a model parameter specific to each random network is varied along the horizontal axis. In general, logistic weight dynamics makes consensus more common but cannot prevent clustering when $R$ is very low.}
    \label{fig:heatmaps with different initial conditions}
\end{figure}

\subsubsection{Random network models} \label{Section: random network models}

Using ER random networks as a baseline, each other model introduces some additional structure. For each type of random network, a version of Figure \ref{fig:exp_logistic_clusters} is given in Appendix \ref{Appendix: Example trajectories} to show example opinion trajectories, network statistics and snapshots of network evolution. Note that, as in Figure \ref{fig:exp_logistic_clusters}, an exponential interaction function is used, and therefore the long-term behaviour follows according to Theorem \ref{Thm: Exp, logistic weight dynamics}. For further details on random network models see, for example, \cite{newman2003structure,goldenberg2010survey}. 

The WS random network creates a `small-world' network by linking individuals mostly to those nearby them (in the sense that the indices $i$ and $j$ are close) with some random rewiring to create connections across the network. \cite{watts1998collective} In the simulations shown in Figure \ref{fig:heatmaps with different initial conditions}, individuals are connected to their $10$ nearest neighbours (in a ring) and the rewiring probability is varied. For Figure \ref{fig:example WS network} in Appendix \ref{Appendix: Example trajectories}, the number of initial connections is increased to $20$, and the rewiring probability is set at $0.05$ to make the network structure more obvious. There is an initial ring structure that rapidly loses many edges, as individuals are unlikely to share similar opinions with their limited number of neighbours. Rather than forming highly connected communities as in Figure \ref{fig:exp_logistic_clusters} (or the later Figure \ref{fig:example SBM network}), the network instead breaks up into a collection of chains in which individuals share similar opinions. 

The BA random network \cite{barabasi1999emergence} is constructed sequentially using preferential attachment: as new nodes are added with a fixed degree, these edges are attached to existing nodes according to the degree of those nodes. This creates central hubs and a power-law degree distribution. In the simulations shown in Figure \ref{fig:heatmaps with different initial conditions}, this initial number of edges is varied. Figure \ref{fig:example BA network} in Appendix \ref{Appendix: Example trajectories} shows an example with $5$ initial edges. There is a clear initial structure of a highly connected centre with fewer connections as one moves away from this. As the network evolves, many individuals on the periphery, who have only a small number of initial connections, become disconnected from the network. Meanwhile, the central group splits into three communities. This leads to a similar opinion clustering, but with many isolated individuals. 

Finally, the SBM adapts the ER network to create networks with some initial community structure \cite{holland1983stochastic}. The population is split into two equally sized communities. Edges within communities exist with probability $p_{\text{within}} = 0.05$, and the probability of edges existing between communities is varied. Figure \ref{fig:example SBM network} in Appendix \ref{Appendix: Example trajectories} shows an example in which $p_{\text{within}}$ is increased to $0.2$ and $p_\text{between} = 0.02$, again to make the network structure more obvious. The motivation for choosing a smaller value of $p_\text{within}$ for simulations in Figure \ref{fig:heatmaps with different initial conditions} is that, similarly to the ER random network, behaviour is less varied when the network is highly connected and interesting transitions occur when these edge probabilities are lower. In Figure \ref{fig:example SBM network}, there are two strong initial communities. By $t=5$ each community has organised into three sub-communities based on opinion. Over time there is a rearrangement from the initial imposed communities to the opinion-driven communities.  

\subsubsection{Opinion clustering under bounded confidence} \label{Section: Simulations BC}

Having introduced the random network models that will be considered, we return to the setting of logistic weight dynamics and a bounded confidence interaction function. To reduce the number of overall parameters, we use here the discontinuous bounded confidence function, which has a single parameter: the confidence radius $R$. This also allows a more direct comparison with the results in \cite{nugent2023evolving} for ER random networks. For each random network structure, we vary both network parameters and the confidence bound $R$, and report the order parameter $Q$ \eqref{eqn: order parameter} once the simulations have reached a steady state (or a maximum time of $t=2000$). Figure \ref{fig:heatmaps with different initial conditions} shows the simulation results. The top row is a reproduction of the results in \cite{nugent2023evolving} for logistic weight dynamics on an ER random network. The remaining rows show new results for the random network models described above. In all cases, $N=500$ and each cell of the heatmap is averaged over $10$ simulations with different initial conditions. 

All four network structures show several common features for both a fixed network (left column) and under logistic weight dynamics (right column). There is always a transition from many clusters (appearing as dark blue) to polarisation (appearing as grey) to consensus (appearing as yellow) as the interaction radius $R$ is increased. As network parameters are varied, the position of these transitions changes. 

For ER, BA and SBM networks, as the network parameter (edge probability, initial degree and between-community ege probability respectively) is increased, the mean degree of the network increases. In these cases, there is a non-monotonic relationship between connectivity and the order parameter for $R$ values around $0.4$. This is seen most strongly in the BA networks. When connectivity is extremely low, consensus does not emerge. This is followed by an optimal region in which consensus occurs, then as connectivity increases further, the population polarises. This effect is seen for both a fixed network and under logistic weight dynamics, which is natural as both will have essentially the same level of network connectivity. This suggests that networks with lower connectivity, in which opinion spread is slower, avoid the rapid formation of clusters and may be more likely to reach consensus. 

For small values of $R$, it appears that the transition from many clusters (blue) to polarisation (grey) is largely unchanged by logistic weight dynamics. This suggests that the high number of clusters is caused by a lack of connectivity in the initial network, which cannot be changed by this type of weight dynamics. 

However, for all four network structures, logistic weight dynamics reduced the parameter space in which polarisation is observed and instead led to consensus. In particular, for WS graph,s this grey region is almost entirely eliminated under logistic weight dynamics. As Figure \ref{fig:example WS network} shows, the network breaks up from an initial ring into a collection of chains. This sharp transition indicates that under logistic weight dynamics either the ring persists (leading to consensus) or breaks into many pieces (leading to a large number of clusters) and that the formation of two large chains/communities is unlikely. This demonstrates that at more moderate $R$ values the change in the strength of weights, even without the addition of new edges, can be effective in creating consensus. 

\section{Early weight dynamics} \label{Section: Early dynamics}

While the study of opinion dynamics, and indeed this paper, typically focuses on the long-term convergence of the system, we may also investigate the early dynamics. In particular, we aim to understand the relationship between the initial distance between individuals' opinions and their corresponding edge weight, with the goal of estimating `typical' edge weight dynamics. As a starting point we consider the simplest possible approximation for memory weight dynamics. 

Denote by $d_{ij} = x_j - x_i$ the distance between two individuals' opinions. As initial opinions are assumed to be ordered we have that $d_{ij} \geq 0$ for $i\leq j$ and focus on these positive values. To reduce notational complexity we denote $\phi_{ij}(t) = \phi(d_{ij}(t))$. 

Assume that for small $t\geq0$, $d_{ij}(t)\approx d_{ij}(0)$. Therefore 
\begin{align*}
    w_{ij}(t) 
    &= w_{ij}(0) \, e^{-t} + \int_0^t e^{s-t} \phi_{ij}(s)\,ds \\
    &\approx w_{ij}(0) \, e^{-t} + \int_0^t e^{s-t} \phi_{ij}(0)\,ds \\
    &= w_{ij}(0) \, e^{-t} + \phi_{ij}(0) (1 - e^{-t}) \,.
\end{align*}
Therefore for short times we have the initial bounds 
\begin{align*}
    \phi_{ij}(0) (1 - e^{-t}) \lesssim w_{ij} \lesssim e^{-t} + \phi_{ij}(s) (1 - e^{-t}) \,.
\end{align*}
However, this completely disregards the opinion formation process. To improve these bounds, we can define 
\begin{align*}
    \quad v_{ij} := \frac{d(d_{ij})}{dt} \bigg|_{t=0} \,,
\end{align*}
and use the approximations
\begin{subequations} \label{eqn: d and phi approximations}
    \begin{align}
        d_{ij}(t) &\approx d_{ij}(0) + t v_{ij} \,,\\
        \phi_{ij}(t) &\approx \phi_{ij}(0) + tv_{ij}\,\phi_{ij}'(0) + \frac{1}{2}t^2v_{ij}^2 \,\phi_{ij}''(0) \,. \label{eqn: phi approximation}
    \end{align}
\end{subequations}
A linear approximation is chosen for $d_{ij}$ since, as will be seen below, it is challenging to approximate the first derivative of $d_{ij}$ and would be impractical to approximate its second derivative due to the high level of coupling. However, using this we may still take a quadratic approximation of $\phi$, provided it is sufficiently smooth. 

For memory weight dynamics this gives the estimate 
\begin{align} \label{eqn: memory estimate}
    w_{ij}(t) 
    &= w_{ij}(0)e^{-t} + \phi_{ij}(0)\,(1-e^{-t}) + v_{ij}\,\phi_{ij}'(0)\,(t - 1 + e^{-t}) + \frac{1}{2}v_{ij}^2\,\phi_{ij}''(0) ((t-2)t + 2 - 2e^{-t}) \,,
\end{align}
while for logistic weight dynamics, we have 
\begin{align}
    w_{ij}(t) &\approx \frac{w_{ij}(0)}{w_{ij}(0) + (1-w_{ij}(0))e^{-\lambda(t)}} \,, \label{eqn:logistic estimate}\\[0.5em]
    \lambda(t) &= 2\Big(\phi_{ij}(0) t + \frac{1}{2}t^2 v_{ij}\,\phi_{ij}'(0) + \frac{1}{6}t^3 v_{ij}^2 \phi_{ij}''(0) \Big) - t \,. \nonumber
\end{align}
Both these estimates require an approximation of $v_{ij}$, the derivative of $d_{ij}$ at time $t=0$. That is, if two individuals begin at a distance $d_{ij}$, how quickly may we expect them to move together/apart. By definition
\begin{align*}
    v_{ij} 
    &= \sum_{n=1}^N \frac{w_{jn}(0)}{k_j(0)} \phi_{jn}(0) \, d_{jn}(0) - \frac{w_{in}(0)}{k_i(0)} \phi_{in}(0) \, d_{in}(0) \,.
\end{align*}

Clearly, this depends on many variables other than the distance $d_{ij}(0)$ - in fact it depends on the initial opinions of the entire population. We therefore aim to estimate the moments of $v_{ij}$. Using the setup considered in \cite{nugent2023evolving}: initial opinions are distributed uniformly at random, the initial network is an ER random network with probability $p\in[0,1]$, with the weight of existing edges chosen uniformly in $[0,1]$. From this we can calculate 
\begin{align*}
    \EE[w_{ij}] = \frac{p}{2} \,, \quad \var[w_{ij}] = p \bigg( \frac{1}{3} - \frac{p}{4} \bigg) \,.
\end{align*}
While division by node degrees introduces some correlation, we use the above to give the following approximations
\begin{align} \label{eqn: weight moment approximations}
    \EE \bigg[ \frac{w_{ij}}{k_i}\bigg] \approx \frac{1}{N} \,, \quad \var \bigg[ \Big(\frac{w_{ij}}{k_i}\Big)^2\bigg] \approx \frac{1}{N^2}\bigg(\frac{4}{3p} - 1 \bigg)\,.
\end{align}
From this, we may estimate the first and second moments of $v_{ij}$, conditional on $d_{ij}$. 

\begin{figure}[ht!]
    \centering
    \includegraphics[width=0.8\linewidth]{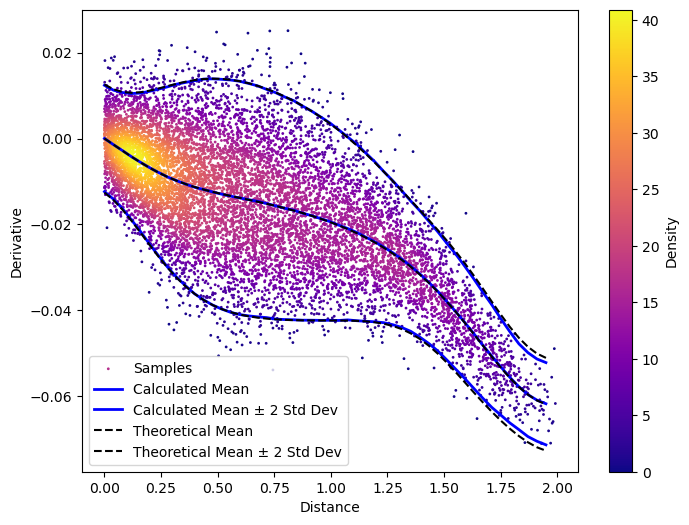}
    \caption{Demonstration of the estimates \eqref{eqn: v_ij estimates} using an exponential interaction function. A population size of $N=100$ was used, with $100$ different, random initial conditions generated to allow accurate calculation of the mean and variance. A representative $10,000$ data points are plotted, coloured using a kernel density estimate. The estimates capture the distribution well.}
    \label{fig:distance vs derivative}
\end{figure}

\begin{proposition} \label{Prop: moments of v_ij}
    Using approximations \eqref{eqn: d and phi approximations} and \eqref{eqn: weight moment approximations}, $v_{ij}$ has the following conditional moments
    \begin{subequations} \label{eqn: v_ij estimates}
    \begin{align}
        \EE[v_{ij}\,|\,d_{ij}] &= \frac{1}{2-d_{ij}}\int_{-1}^{1-d_{ij}} I_1(d_{ij},x) -\frac{2}{N} \Big( I_1(d_{ij},x) + \phi_{ij} \, d_{ij} \Big)\, dx \,,\\
        \EE[v_{ij}^2\,|\,d_{ij}] 
        &= \frac{1}{2-d_{ij}}\int_{-1}^{1-d_{ij}} \frac{1}{N^2}f(x) + \frac{1}{N^2}\bigg(\frac{4}{3p} - 1 \bigg) g(x) \, dx \,, \label{eqn: v_ij variance estimate}
    \end{align} 
    \end{subequations}
    where
    \begin{align*}
        I_m(d_{ij},x) 
        &= \frac{1}{2}\int_{-1}^1 \Big( \phi(y-x-d_{ij}) \, (y-x-d_{ij}) - \phi(y-x) \, (y-x) \Big)^m\,dy \,, \\[0.5em]
        f(d_{ij},x) &= (N-2)(N-3) I_1(d_{ij},x)^2 + (N-2) I_2(d_{ij},x) + 4(\phi_{ij}d_{ij})^2 - 4(\phi_{ij}\,d_{ij})(N-2)I_1(d_{ij},x) \,,\\[0.5em]
        g(d_{ij},x) &= 2(\phi_{ij} \, d_{ij} )^2 \, + (N-2) \big(J(x + d_{ij}) + J(x) \big) \,,\\[0.5em]
        J(x) &= \frac{1}{2} \int_{-1}^1 \Big(\phi(y-x)\,(y-x) \Big)^2 \, dy \,.
    \end{align*}
\end{proposition}
\begin{proof}
    See Appendix \ref{Appendix: Proofs}.
\end{proof}

An example of these approximations for an exponential interaction function and $N=100$ are shown in Figure \ref{fig:distance vs derivative}. The mean and variance estimates fit well and the mean $\pm$ two standard deviations provides reasonable upper and lower bounds on $v_{ij}$. There is a slight disagreement in variance for higher values of $d_{ij}$, although the data in this area is naturally more sparse. 

The precise shape of these curves will be dependent on the population size, the interaction function $\phi$ and the edge probability $p$, which determines the variance of edges and therefore the balance of the functions $f$ and $g$ in \eqref{eqn: v_ij variance estimate}. Note that $g$ contains no terms of order $N^2$. Therefore the variance arising from uncertainty in the network, given by the term
\begin{align*}
    \frac{1}{N^2} \bigg( \frac{4}{3p} - 1 \bigg) g(x) \,,
\end{align*}
in $\EE[v_{ij}^2 \, | \, d_{ij}]$ will vanish as $N\rightarrow\infty$. However uncertainty due to the unknown position of $x_i$ remains as $f$ contains terms of order $N^2$. This seems natural as the ER random network is known to converge to a uniform graphon as $N\rightarrow\infty$ \cite{lovasz2012large}, meaning the behaviour of the network becomes deterministic and therefore would contribute no uncertainty to this approximation. 

\begin{figure}[ht!]
    \centering
    \includegraphics[width=.9\linewidth]{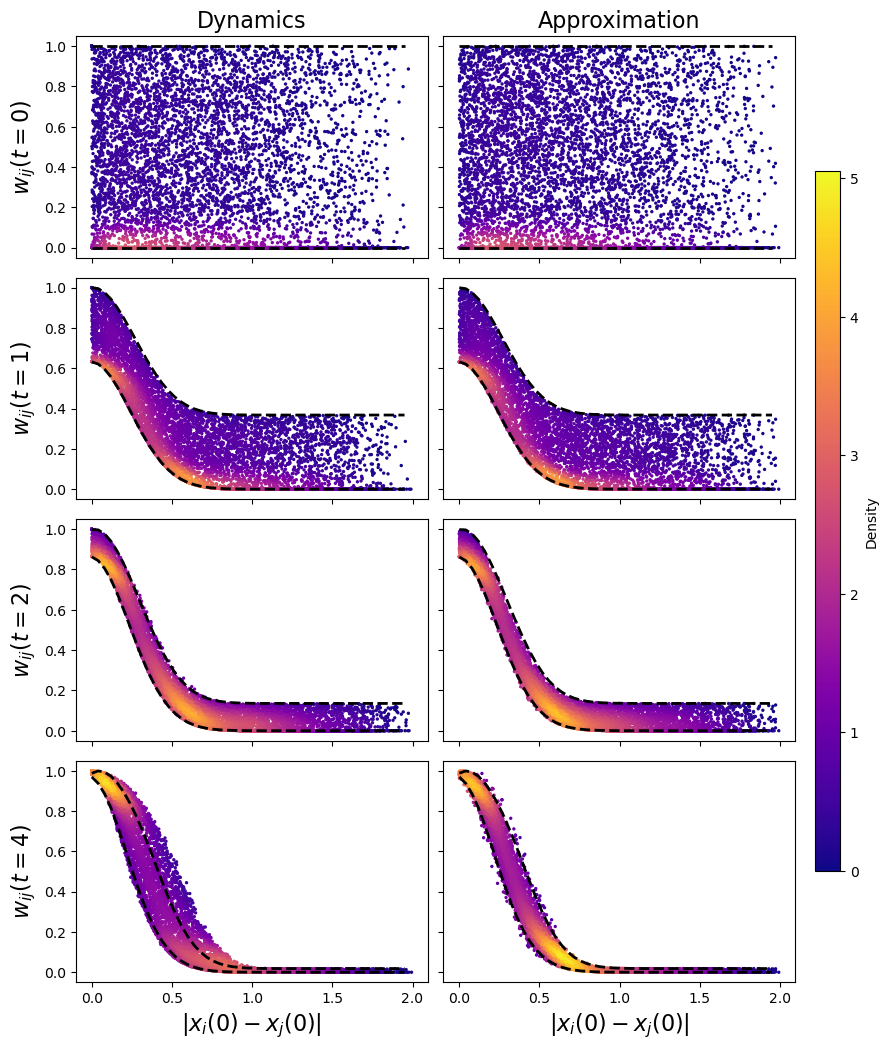}
    \caption{Weights under memory weight dynamics (left) and approximation  for small $t$ obtained by combining \eqref{eqn: d and phi approximations} and \eqref{eqn: memory estimate} (right). The scatter plot shows actual simulation data on the left and an equal number of randomly sampled points on the right. Both are coloured using a kernel density estimate. Black dashed lines show upper and lower estimates using $w_{ij}=0,1$ and the mean of $v_{ij}$ plus or minus two standard deviations (using the estimates from Proposition \ref{Prop: moments of v_ij}). There is good agreement between simulations and approximations until $t=2$, after which the dynamics deviate from each other.}
    \label{fig:weight dynamics short time estimates}
\end{figure}

\begin{figure}[ht!]
    \centering
    \includegraphics[width=.9\linewidth]{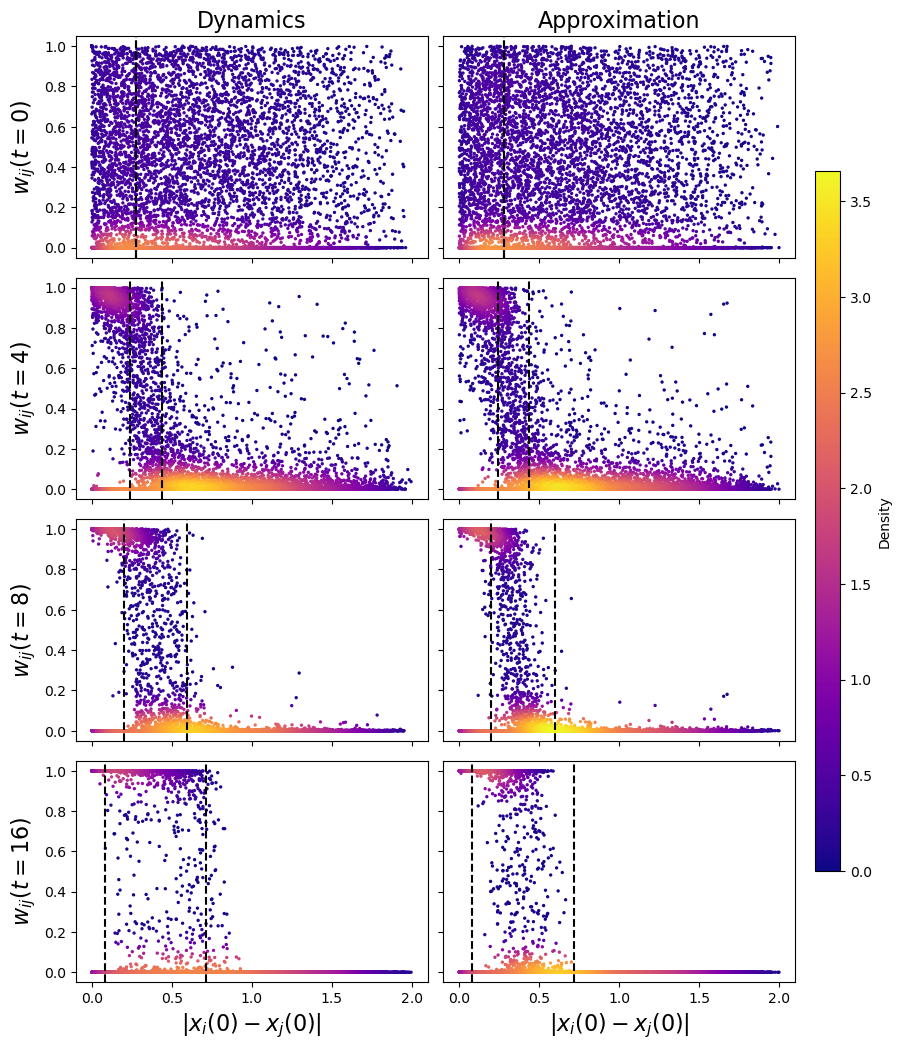}
    \caption{Weights under logistic weight dynamics (left) and approximation  for small $t$ obtained by combining \eqref{eqn: d and phi approximations} and \eqref{eqn:logistic estimate} (right). The scatter plot shows actual simulation data on the left and an equal number of randomly sampled points on the right. Both are coloured using a kernel density estimate. Black dashed lines show the boundaries $d^\pm$, described at the end of Section \ref{Section: Early dynamics}.}
    \label{fig:early weight dynamics logistic}
\end{figure}

These estimates can then be combined with \eqref{eqn: memory estimate} or \eqref{eqn:logistic estimate} to approximate $w_{ij}$. It should be noted that the approximation of $d_{ij}$ is linear, therefore we cannot account for possible non-linear effects in the opinion dynamics that are likely to occur at later times. An example of this, using the same setup as for Figure \ref{fig:distance vs derivative} is shown in Figure \ref{fig:weight dynamics short time estimates}. Specifically, a random Erdos-Renyi network with $N=100$, $p=0.5$ is created, existing edges are given a random weight uniformly in $[0,1]$ and initial opinions are chosen uniformly in $[-1,1]$. The left column then shows the weight $w_{ij}(t)$ at four timepoints for 10,000 randomly selected pairs. Their positions on the $x$-axis are given by their initial distance $d_{ij}(0)$ to allow for a comparison against the approximations above. Points are coloured according to a Gaussian kernel density estimate (KDE) to show their distribution. 

For the right hand column 10,000 random samples are generated for each timepoint. For each sample an initial edge weight is chosen as $w_{ij}(0)=0$ with probability $1-p$ or $w_{ij}\sim \text{Uniform}[0,1]$ with probability $p$. Opinions $x_i$ and $x_j$ are sampled uniformly in $[-1,1]$ and swapped if necessary so that $d_{ij} = x_j - x_i \geq 0$. Given this $d_{ij}$, $v_{ij}$ is sampled from a normal distribution with mean $\mathds{E}[v_{ij}\,|\,d_{ij}]$ and variance $\var[v_{ij}\,|\,d_{ij}]$ given by \eqref{eqn: v_ij estimates}. The estimate \eqref{eqn: memory estimate} is then used to approximate $w_{ij}(t)$. Again the points are coloured using a KDE to show their distribution.

Finally a reasonable upper and lower bound are shown in black dashed lines on both columns. To calculate these we require upper and lower bounds for $w_{ij}$ and $v_{ij}$, so take $1$ and $0$ respectively for $w_{ij}$ and the mean $\pm$ two standard deviations for $v_{ij}$ (as shown by the black dashed lines in Figure \ref{fig:distance vs derivative}).

Due to the independence between individuals both the dynamics and the approximation have, by construction, the same initial distribution over $(d_{ij},w_{ij})$. As $p=0.5$, half the weights in the network will begin at $0$ with the other half being distributed across $[0,1]$. For distances $d_{ij}\geq0$, their distribution is given by
\begin{align*}
    \mathbb{P}[d_{ij} =d] 
    &= \frac{1}{2} \int_{-1}^1 \mathbb{P}[d_{ij} =d \, | \, x_i = x] \, dx 
    = \frac{1}{2} \int_{-1}^1 \mathbb{P}[x_j = x +d \, | \, x_i = x] \, dx 
    = \frac{1}{2} \bigg(1 - \frac{d}{2} \bigg) \,.
\end{align*}

The dynamics and approximation continue to match well until $t=2$, after which they begin to diverge. As stated above, this is to be expected from a small-time approximation. For $t=4$ weights in the dynamics are higher than expected, indicating that pairs are closer than predicted by the linear distance approximation, leading to higher values of $\phi$ and subsequently higher weights. This error is most likely caused by the error in the linear approximation of $d_{ij}$ propagating through our estimates. 

Looking at $t\leq2$ we see that in both dynamics and the approximation the distribution is concentrated near the lower bound as the majority of weights begin at $w_{ij}(0)=0$. However, we see the formation of two regions of higher density near $d_{ij}=0$ and slightly above $d_{ij}=0.5$. The first corresponds to pairs of individuals who are initially close in opinion but disconnected in the network, hence they rapidly form a stronger connection. The second corresponds to pairs who are initially too far away to interact but over time move closer and form a weak connection. As the $x$-values are given by $d_{ij}(0)$ their distribution does not change, hence the increased density indicates a small range of possible weights for that value of $d_{ij}(0)$, meaning greater predictability. Distances in the range $d_{ij}\in [0.25,0.5]$ do not show this increase in density, indicating that their edge weight may plausibly take a broader range of values. In the approximation these higher-density regions become more pronounced at $t=4$, while in the dynamics only the region near $d_{ij}=0$ persists. This is likely due to non-linearity in the dynamics which reduces the predictability of the network at later timepoints. 

In both the dynamics and approximation, all edges are effectively removed for values $d_{ij}\geq 1$ by $t=4$, indicating the formation of communities within the network, similarly to Figure \ref{fig:exp_logistic_clusters}. However, in this case we have memory weight dynamics rather than logistic weight dynamics and from Corollary \ref{Cor: Exp, memory weight dynamics} we know that eventually the population must reach consensus and $w_{ij}\rightarrow \phi(0)=1$. Hence, these initially removed edges must reappear at a later time. This long-time behaviour is quite different from the behaviour of the short-time approximation shown in Figure \ref{fig:weight dynamics short time estimates}, indicating an initial clustering phase, accompanied by a corresponding formation of communities in the network, before a final convergence to consensus. 

Results for logistic weight dynamics (using the same interaction function and setup for initial conditions) are shown in Figure \ref{fig:early weight dynamics logistic}, where there is good agreement even at later times. In this case $w_{ij}=0,1$ cannot be used to form upper and lower bounds as these are fixed points of the dynamics. However, using $v_{ij} = \EE[v_{ij}] \pm 2\sqrt{\var{[v_{ij}]}}$ we can give reasonable upper and lower estimates for $\phi_{ij}(t)$ using \eqref{eqn: phi approximation}. These are denoted $\phi_{ij}^\pm(t)$. From these we can find the distances $d^\pm(t)$ at which $\phi^\pm(d(t))=1/2$ to separate regions of increasing and decreasing weights. For $d_{ij}\leq d^-(t)$ distances are small so any weight that is initially non-zero will have increased to near 1. For $d_{ij}\geq d^+(t)$ initial distances are too great and edges with $w_{ij}<1$ will have decreased towards 0. For $d_{ij} \in (d^-(t),d^+(t))$ we expect edges to be in transition. This behaviour is exactly what is observed in Figure \ref{fig:early weight dynamics logistic}. 

In both cases we are limited here to predicting individual weights. Expanding the approach to describe the evolution of broader network structures would remove a currently central assumption of independence between edges, which is present at $t=0$ but lost over time. 

\section{Conclusion} \label{Section: Conclusion}

In this paper, we examined both the long-term and short-term behaviour of the dynamic network model \eqref{Eqn: general dynamic network system}. When the interaction function is strictly positive over $[0,2]$, the convergence behaviour of the model is shown to match that of the fixed network case, with the formation of distinct opinion clusters, except that with a dynamic network, there is also a corresponding emergence of community structure. In the case of memory weight dynamics, we show convergence to consensus and a fully connected network, although the exploration of short-term dynamics in Section \ref{Section: Early dynamics} shows that there is also an initial phase of clustering in both opinions and the edge weights. 

When the interaction function does not have full support, as with the smoothed bounded confidence function, the picture is more complex. For memory weight dynamics a similar convergence to clustering result holds, while for logistic weight dynamics, the strong dependence on initial conditions prevents the same analytical approach from being applied. A key factor is that network dynamics remove the requirement that clusters must lie a distance apart where the interaction function is zero, as highlighted by the examples shown in Figure \ref{fig:diagram} and Figure \ref{fig:diagram 2}. This inability to locate clusters creates further dependence on the initial network structure. The impact of this structure was explored through simulations in Section \ref{Section: Simulations BC} using different random network models. Although the precise behaviours differed, all models showed a transition from clustering to polarisation to consensus, with network parameters affecting the location of these transitions. Moreover, logistic weight dynamics universally reduced polarisation, although it had little effect when the interaction function had a very small radius.

The various approaches in this paper provide a more complete understanding of both the short- and long-time dynamics of the coupled system \eqref{Eqn: general dynamic network system} under a variety of interaction functions and network dynamics. In general, the adaptive network drives the formation of communities that mirror opinion clusters, but the precise behaviour is highly dependent on the support of the interaction function and the ability of the weight dynamics to create new edges. 

\newpage
\bibliography{bibliography}
\bibliographystyle{unsrt}

\section*{Acknowledgments}
AN was supported by Engineering and Physical Sciences Research Council through the Mathematics of Systems II Centre for Doctoral Training at the University of Warwick (reference EP/S022244/1). 

For the purpose of open access, the authors have applied a Creative Commons Attribution (CC-BY) license to Any Author Accepted Manuscript version arising from this submission.

The authors would like to thank Sasha Glendinning for their assistance in proof reading the manuscript. 

\newpage
\appendix

\section{Proofs} \label{Appendix: Proofs}

Proof of Proposition \ref{Prop: moments of v_ij}
\begin{proof}
    We begin with the first moment of $v_{ij}$,
    \begin{align}
        \EE[v_{ij}\,|\,d_{ij}] &= \int_{-1}^{1} \EE[v_{ij}\,|\,d_{ij}, x_i = x] \, \mathbb{P}[x_i = x | d_{ij}] \, dx \nonumber\\
        &= \frac{1}{2-d_{ij}}\int_{-1}^{1-d_{ij}} \EE[v_{ij}\,|x_i = x, x_j = x+d_{ij}] \, dx \nonumber\\
        &= \frac{1}{N(2-d_{ij})}\int_{-1}^{1-d_{ij}} \EE\Bigg[\sum_{n=1}^N \phi_{jn}(0) \, d_{jn}(0) - \phi_{in}(0) \, d_{in}(0)\,\bigg|x_i = x, x_j = x+d_{ij}\Bigg] \, dx \,. \label{eqn: First moment of v_ij (1)}
    \end{align}
    The sum inside the expectation is then separated into terms involving $i,j$ and other independent terms, 
    \begin{align*}
        \frac{1}{N}\EE\Bigg[\sum_{n=1}^N \phi_{jn}(0) \, d_{jn}(0) - \phi_{in}(0) \, d_{in}(0)\Bigg]
        &= \frac{1}{N}\EE\Bigg[\sum_{n=1}^N \phi_{jn} \, d_{jn} - \phi_{in} \, d_{in}\Bigg] \\
        &= \frac{1}{N}\EE\Bigg[\sum_{n\neq i,j}^N \phi_{jn} \, d_{jn} - \phi_{in} \, d_{in}\Bigg] \,.
    \end{align*}
    Recall that $d_{ii} = d_{jj} = 0$, so we have 
    \begin{align*}
    \frac{1}{N}\EE\Bigg[\sum_{n=1}^N \phi_{jn}(0) \, d_{jn}(0) - \phi_{in}(0) \, d_{in}(0)\Bigg]
        &+ \frac{1}{N}\EE\Bigg[\phi_{ji} \, d_{ji} - \phi_{ii} \, d_{ii} + \phi_{jj} \, d_{jj} - \phi_{ij} \, d_{ij}\Bigg] \\
        &= \Big(1-\frac{2}{N}\Big) \EE[\phi_{jn} \, d_{jn} - \phi_{in} \, d_{in}] \\
        &+ \frac{1}{N}\EE\Bigg[\phi_{ji} \, d_{ji} - \phi_{ij} \, d_{ij}\Bigg] \\
        &= \Big(1-\frac{2}{N}\Big) I_1(d_{ij},x) - \frac{2}{N}\phi_{ij} \, d_{ij}\,.
    \end{align*}
    So \eqref{eqn: First moment of v_ij (1)} becomes
    \begin{align*}
        \EE[v_{ij}\,|\,d_{ij}] &= \frac{1}{2-d_{ij}}\int_{-1}^{1-d_{ij}} \Big(1-\frac{2}{N}\Big) I_1(d_{ij},x) - \frac{2}{N}\phi_{ij} \, d_{ij} \, dx \,.
    \end{align*}
    Next we consider the second moment.
    \begin{align*}
        \EE[v_{ij}^2\,|\,d_{ij}] 
        &= \int_{-1}^{1} \EE[v_{ij}^2\,|\,d_{ij}, x_i = x] \, \mathbb{P}[x_i = x | d_{ij}] \, dx \\
        &= \frac{1}{2-d_{ij}}\int_{-1}^{1-d_{ij}} \EE[v_{ij}^2\,|x_i = x, x_j = x+d_{ij}] \, dx \\
        &= \frac{1}{(2-d_{ij})}\int_{-1}^{1-d_{ij}} \EE\Bigg[\bigg(\sum_{n=1}^N \frac{w_{jn}}{k_j}\phi_{jn}(0) \, d_{jn}(0) - \frac{w_{in}}{k_i}\phi_{in}(0) \, d_{in}(0)\bigg)^2\,\bigg|x_i = x, x_j = x+d_{ij}\Bigg] \, dx \,.
    \end{align*} 
    To simplify notation we write $\omega_{ij} = w_{ij}/k_i$. 
    
    This time more cases are required. We consider each of the following:
    \begin{enumerate}
        \item $n=i$, $m=1,\dots,N$
        \item $n=j$, $m=1,\dots,N$
        \item $n\neq i,j$, $m=i$, 
        \item $n\neq i,j$, $m=j$, 
        \item $n\neq i,j$, $m=n$, 
        \item $n\neq i,j$, $m\neq i,j,n$. 
    \end{enumerate}
    
    \textbf{Case 1:} ($n=i$, $m=1,\dots,N$)
    
    \begin{align*}
        \EE\Bigg[&\sum_{m=1}^N (\omega_{ji}\phi_{ji} \, d_{ji} )(\omega_{jm}\phi_{jm} \, d_{jm} - \omega_{im}\phi_{im} \, d_{im}) \Bigg] \\
        &= \EE\Bigg[(\omega_{ji}\phi_{ji} \, d_{ji} )^2 
        - (\omega_{ji}\phi_{ji} \, d_{ji} )(\omega_{ij}\phi_{ij} \, d_{ij}) 
        + \sum_{m\neq i,j} (\omega_{ji}\phi_{ji} \, d_{ji} )(\omega_{jm}\phi_{jm} \, d_{jm} - \omega_{im}\phi_{im} \, d_{im}) \Bigg] \\
        &= (\phi_{ij} \, d_{ij} )^2\EE[\omega_{ji}^2 ] 
        + (\phi_{ij} \, d_{ij} )^2 \EE[\omega_{ji} \, \omega_{ij} ]
        - (\phi_{ij} \, d_{ij}) \frac{1}{N} \EE\Bigg[\sum_{m\neq i,j} (\omega_{jm}\phi_{jm} \, d_{jm} - \omega_{im}\phi_{im} \, d_{im}) \Bigg] \\
        &= (\phi_{ij} \, d_{ij} )^2\EE[\omega^2 ] 
        + (\phi_{ij} \, d_{ij} )^2 \frac{1}{N^2}
        - (\phi_{ij} \, d_{ij}) \frac{1}{N^2} \EE\Bigg[\sum_{m\neq i,j} (\phi_{jm} \, d_{jm} - \phi_{im} \, d_{im}) \Bigg]  \,.
    \end{align*}
    In the first line, we separate out the terms $m=i,j$ as these are not independent. In the second line, the known values of $\phi_{ij}\,d_{ij}$ are removed from the expectations. In the third line, the independence of $\omega_{ij}$'s is used. 

    Next we use that $\EE[\omega]=N^{-1}$ to write $\EE[\omega^2]=\var[\omega] + N^{-2}$. In addition the $m\neq i,j$ terms are collected as $(N-2)$ independent copies of $I_1(d_{ij},x)$. This gives
    \begin{align*}
        \EE\Bigg[&\sum_{m=1}^N (\omega_{ji}\phi_{ji} \, d_{ji} )(\omega_{jm}\phi_{jm} \, d_{jm} - \omega_{im}\phi_{im} \, d_{im}) \Bigg] \\
        &= (\phi_{ij} \, d_{ij} )^2 \bigg( \var[w] + \frac{1}{N^2} \bigg) 
        + (\phi_{ij} \, d_{ij} )^2 \frac{1}{N^2}
        - (\phi_{ij} \, d_{ij}) \frac{N-2}{N^2} I_1(d_{ij},x) \\
        &= (\phi_{ij} \, d_{ij} )^2 \, \var[w]  
        + \frac{\phi_{ij} \, d_{ij}}{N^2} \Bigg( 2(\phi_{ij} \, d_{ij} )
        - (N-2) I_1(d_{ij},x) \Bigg) \,.
    \end{align*}
    Here, the last line is a simple rearrangement of terms. 
    
    \textbf{Case 2:} ($n=j$, $m=1,\dots,N$)
    
    This is the same as Case 1. 
    \begin{align*}
        & \EE\Bigg[\sum_{m=1}^N (\phi_{jj} \, d_{jj} - \phi_{ij} \, d_{ij})(\phi_{jm} \, d_{jm} - \phi_{im} \, d_{im}) \Bigg] \\
        &= (\phi_{ij} \, d_{ij} )^2 \, \var[w]  
        + \frac{\phi_{ij} \, d_{ij}}{N^2} \Bigg( 2(\phi_{ij} \, d_{ij} )
        - (N-2) I_1(d_{ij},x) \Bigg) \,.
    \end{align*}

    \textbf{Case 3:} ($n\neq i,j$, $m=i$)
    
    Due to the symmetry between $n$ and $m$ this will give
    \begin{align*}
        & \EE\Bigg[\sum_{n\neq i,j} (\phi_{jn} \, d_{jn} - \phi_{in} \, d_{in})(\phi_{ji} \, d_{ji} - \phi_{ii} \, d_{ii}) \Bigg] \\
        &= (\phi_{ij} \, d_{ij} )^2 \, \var[w]  
        - \frac{\phi_{ij} \, d_{ij}}{N^2} (N-2) I_1(d_{ij},x) \,.
    \end{align*}
    In this case the $n=i,j$ terms have already been removed so do not need to be accounted for separately. 
    
    \textbf{Case 4:} ($n\neq i,j$, $m=j$)
    
    This is the same as Case 3. 
    \begin{align*}
        & \EE\Bigg[\sum_{n\neq i,j} (\phi_{jn} \, d_{jn} - \phi_{in} \, d_{in})(\phi_{jj} \, d_{jj} - \phi_{ij} \, d_{ij}) \Bigg] \\
        &= (\phi_{ij} \, d_{ij} )^2 \, \var[w]  
        - \frac{\phi_{ij} \, d_{ij}}{N^2} (N-2) I_1(d_{ij},x) \,.
    \end{align*}

    \textbf{Case 5:} ($n\neq i,j$, $m=n$)
    
    \begin{align*}
        E &:= \EE\Bigg[\sum_{n\neq i,j} (\omega_{jn}\phi_{jn} \, d_{jn} - \omega_{in}\phi_{in} \, d_{in})(\omega_{jn}\phi_{jn} \, d_{jn} - \omega_{in}\phi_{in} \, d_{in}) \Bigg] \\
        &= \EE\Bigg[\sum_{n\neq i,j} (\omega_{jn}\phi_{jn} \, d_{jn} - \omega_{in}\phi_{in} \, d_{in})^2 \Bigg] \\
        &= (N-2) \EE\Big[(\omega_{jn}\phi_{jn} \, d_{jn} - \omega_{in}\phi_{in} \, d_{in})^2 \Big] \,.
    \end{align*}
    From here, we expand the square inside the expectation and use the independence between $\omega$'s and opinions to obtain
    \begin{align*}
        E&= (N-2) \bigg( \EE[\omega^2] \EE\Big[(\phi_{jn} \, d_{jn})^2 \Big]  - 2 \frac{1}{N^2} \EE\Big[(\phi_{jn} \, d_{jn}) (\phi_{in} \, d_{in})\Big] + \EE[\omega^2] \EE\Big[(\phi_{in} \, d_{in})^2 \Big] \bigg) \,.
    \end{align*}
    Note that although $d_{ij}$ (and consequently $\phi_{ij}$) is known, $\phi_{jn} d_{jn}$ and $\phi_{in} d_{in}$ are still random variables and their distributions are different (since $x_j$ and $x_i$ are different). Again using $\EE[\omega^2]=\var[\omega] + N^{-2}$ and rearranging gives 
    \begin{align*}
        E &= (N-2) \bigg( \Big(\var[w]+\frac{1}{N^2}\Big) \EE\Big[(\phi_{jn} \, d_{jn})^2 + (\phi_{in} \, d_{in})^2  \Big]  - 2 \frac{1}{N^2} \EE\Big[(\phi_{jn} \, d_{jn}) (\phi_{in} \, d_{in})\Big] \bigg) \\
        &= (N-2) \bigg( \var[w] \EE\Big[(\phi_{jn} \, d_{jn})^2 + (\phi_{in} \, d_{in})^2  \Big] + \frac{1}{N^2} \EE\Big[(\phi_{jn} \, d_{jn})^2 + (\phi_{in} \, d_{in})^2  -2(\phi_{jn} \, d_{jn}) (\phi_{in} \, d_{in})\Big] \bigg) \,.
    \end{align*}
    The second term can now again be written as a square, giving 
    \begin{align*}
        E &= (N-2) \bigg( \var[w] \EE\Big[(\phi_{jn} \, d_{jn})^2 + (\phi_{in} \, d_{in})^2 \Big] + \frac{1}{N^2} \EE\Big[(\phi_{jn} \, d_{jn} -\phi_{in} \, d_{in})^2 \Big] \bigg) \\
        &= (N-2) \var[w] \Big(J(x + d_{ij}) + J(x) \Big) + \frac{N-2}{N^2} I_2(d_{ij},x) \,,
    \end{align*}
    where
    \begin{align*}
        J(x) = \frac{1}{2} \int_{-1}^1 \Big(\phi(y-x)\,(y-x) \Big)^2 \, dy \,.
    \end{align*}
    
    \textbf{Case 6:} ($n\neq i,j$, $m\neq i,j,n$)
    
    When $m\neq n,i,j$ the random variables $(\omega_{jn} \phi_{jn} \, d_{jn} - \omega_{in} \phi_{in} \, d_{in})$ and $(\omega_{jm} \phi_{jm} \, d_{jm} - \omega_{im} \phi_{im} \, d_{im})$ are immediately uncorrelated. Therefore, there is no need to separate any term,s and we may calculate directly
    \begin{align*}
        &\EE\Bigg[\sum_{n\neq i,j} \sum_{m\neq n,i,j} ( \omega_{jn} \phi_{jn} \, d_{jn} - \omega_{in} \phi_{in} \, d_{in})( \omega_{jm} \phi_{jm} \, d_{jm} - \omega_{im} \phi_{im} \, d_{im}) \Bigg] \\
        &= \sum_{n\neq i,j} \sum_{m\neq n,i,j}\EE[\omega_{jn} \phi_{jn} \, d_{jn} - \omega_{in} \phi_{in} \, d_{in}]\,\EE[\omega_{jm} \phi_{jm} \, d_{jm} - \omega_{im} \phi_{im} \, d_{im}] \\
        &= \frac{1}{N^2}\sum_{n\neq i,j} \sum_{m\neq n,i,j} I_1(d_{ij},x)^2 \\
        &= \frac{(N-2)(N-3)}{N^2} I_1(d_{ij},x)^2 \,.
    \end{align*}
    
    Putting all this together, we have
    \begin{align*}
        &\EE\Bigg[ \bigg(\sum_{n=1}^N \phi_{jn} \, d_{jn} - \phi_{in} \, d_{in}\bigg) \bigg(\sum_{n=1}^N \phi_{jn} \, d_{jn} - \phi_{in} \, d_{in}\bigg) \Bigg] \\
        &= 2(\phi_{ij} \, d_{ij} )^2 \, \var[w]  
        + 2\frac{\phi_{ij} \, d_{ij}}{N^2} \Bigg( 2(\phi_{ij} \, d_{ij} )
        - (N-2) I_1(d_{ij},x) \Bigg) \\
        &+ 2(\phi_{ij} \, d_{ij} )^2 \, \var[w]  
        - 2\frac{\phi_{ij} \, d_{ij}}{N^2} (N-2) I_1(d_{ij},x) \\
        &+  (N-2) \var[w] \Big(J(x + d_{ij}) + J(x) \Big) + \frac{1}{N^2} (N-2) I_2(d_{ij},x) \\
        &+ \frac{1}{N^2} (N-2)(N-3) I_1(d_{ij},x)^2 \\
        &= \Big( 4(\phi_{ij} \, d_{ij} )^2 \, + (N-2) \big(J(x + d_{ij}) + J(x) \big) \Big) \var[w] \\
        &+  \frac{4\phi_{ij} \, d_{ij}}{N^2} \Big( \phi_{ij} \, d_{ij} 
        - (N-2) I_1(d_{ij},x)\Big) + \frac{1}{N^2} (N-2) I_2(d_{ij},x) + \frac{1}{N^2} (N-2)(N-3) I_1(d_{ij},x)^2 \,.
    \end{align*}
    
    Denoting 
    \begin{align*}
        f(d_{ij},x) &= (N-2)(N-3) I_1(d_{ij},x)^2 + (N-2) I_2(d_{ij},x) + 4(\phi_{ij}d_{ij})^2 - 4(\phi_{ij}d_{ij})(N-2)I_1(d_{ij},x) \,,\\
        g(d_{ij},x) &= 4(\phi_{ij} \, d_{ij} )^2 \, + (N-2) \big(J(x + d_{ij}) + J(x) \big) \,,
    \end{align*}
    we have altogether 
    \begin{align*}
        \EE[v_{ij}\,|\,d_{ij}] &= \frac{1}{2-d_{ij}}\int_{-1}^{1-d_{ij}} I_1(d_{ij},x) -\frac{2}{N} \Big( I_1(d_{ij},x) + \phi_{ij} \, d_{ij} \Big)\, dx \,,\\
        \EE[v_{ij}^2\,|\,d_{ij}] 
        &= \frac{1}{2-d_{ij}}\int_{-1}^{1-d_{ij}} \frac{1}{N^2}f(x) + \var[\omega] g(x) \, dx \,.
    \end{align*} 
    From \eqref{eqn: weight moment approximations} we have
    \begin{align*}
        \var[\omega] \approx \frac{4-3p}{3 N^2 p} \,,
    \end{align*}
    hence
    \begin{align*}
        \EE[v_{ij}\,|\,d_{ij}] &= \frac{1}{2-d_{ij}}\int_{-1}^{1-d_{ij}} I_1(d_{ij},x) -\frac{2}{N} \Big( I_1(d_{ij},x) + \phi_{ij} \, d_{ij} \Big)\, dx \,,\\
        \EE[v_{ij}^2\,|\,d_{ij}] 
        &= \frac{1}{2-d_{ij}}\int_{-1}^{1-d_{ij}} \frac{1}{N^2}f(x) + \frac{4-3p}{3pN^2} g(x) \, dx \,.
    \end{align*} 
\end{proof}

\section{Example trajectories} \label{Appendix: Example trajectories}

Below are several versions of Figure \ref{fig:exp_logistic_clusters} using different random network models to generate $w(0)$. As in Figure \ref{fig:exp_logistic_clusters}, an exponential interaction function is used. A detailed exploration of the dynamics under bounded confidence interaction functions can be found in Section \ref{Section: Simulations BC}.


\begin{figure}[ht!]
    \centering
    \begin{subfigure}{\textwidth}
        \centering
        \includegraphics[width=0.8\linewidth]{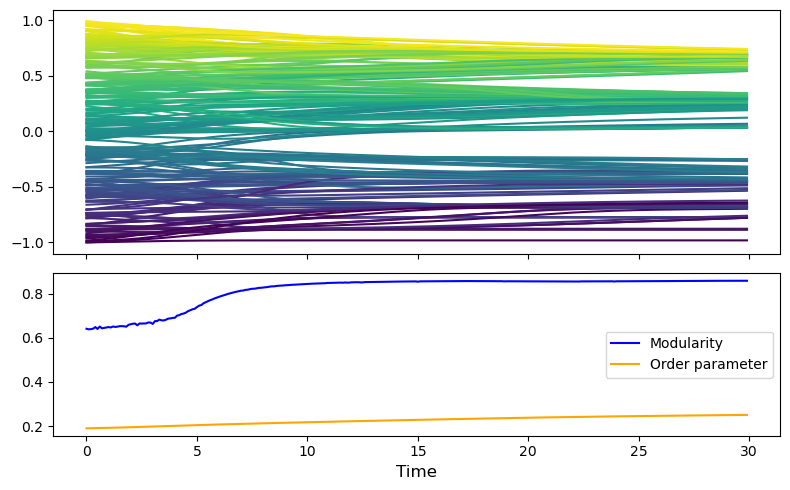}
        \caption{The top panel shows opinion trajectories over time, coloured according to initial opinion. Although some opinion clusters form, there are many individuals who become isolated. The bottom panel shows the order parameter and network modularity over time. The modularity is initially high and rises further at the start of the dynamics as the network removes many edges.}
    \end{subfigure}
    \begin{subfigure}{\textwidth}
        \centering
        \includegraphics[width=0.75\linewidth]{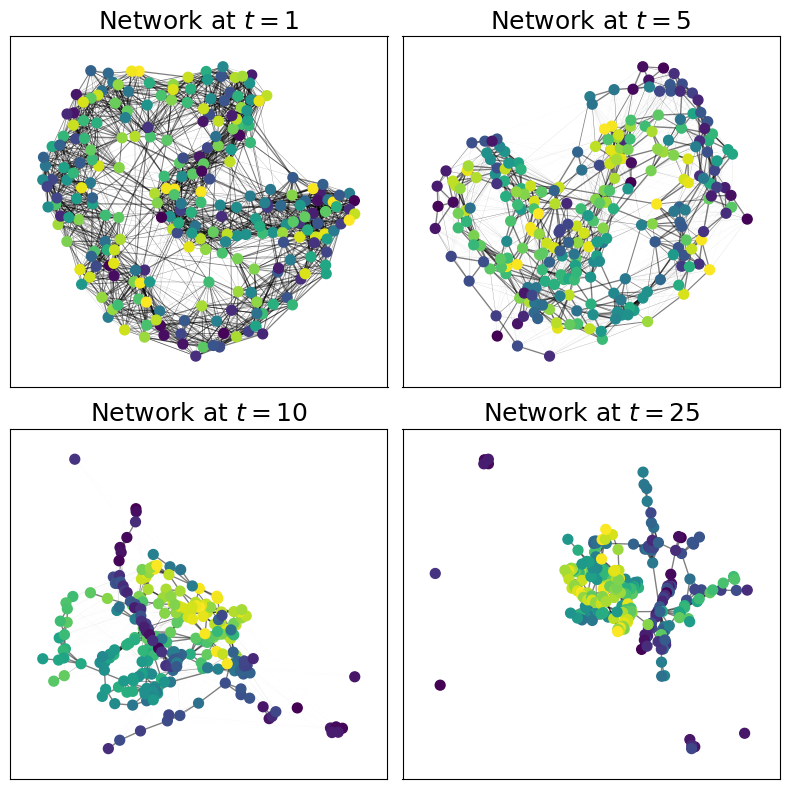}
        \caption{Snapshots of network structure. Edge width is drawn according to edge weight; nodes are coloured according to individuals' initial opinions. The network's initial ring structure is broken down into a collection of chains.}
    \end{subfigure}
    \caption{Demonstration of opinion formation under logistic weight dynamics using a \textbf{WS initial network}. A complete description of the setup can be found in Section \ref{Section: random network models}.}
    \label{fig:example WS network}
\end{figure}


\begin{figure}[ht!]
    \centering
    \begin{subfigure}{\textwidth}
        \centering
        \includegraphics[width=0.8\linewidth]{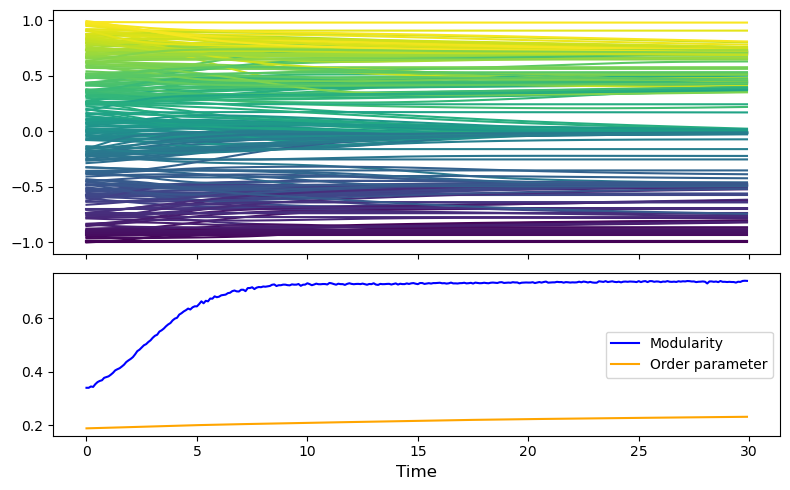}
        \caption{The top panel shows opinion trajectories over time, coloured according to initial opinion. Many individuals become isolated, and there is minimal opinion clustering. The bottom panel shows the order parameter and network modularity over time. The modularity rises continuously from $t=0$ to approximately $t=6$ at which point it levels out.}
    \end{subfigure}
    \begin{subfigure}{\textwidth}
        \centering
        \includegraphics[width=0.75\linewidth]{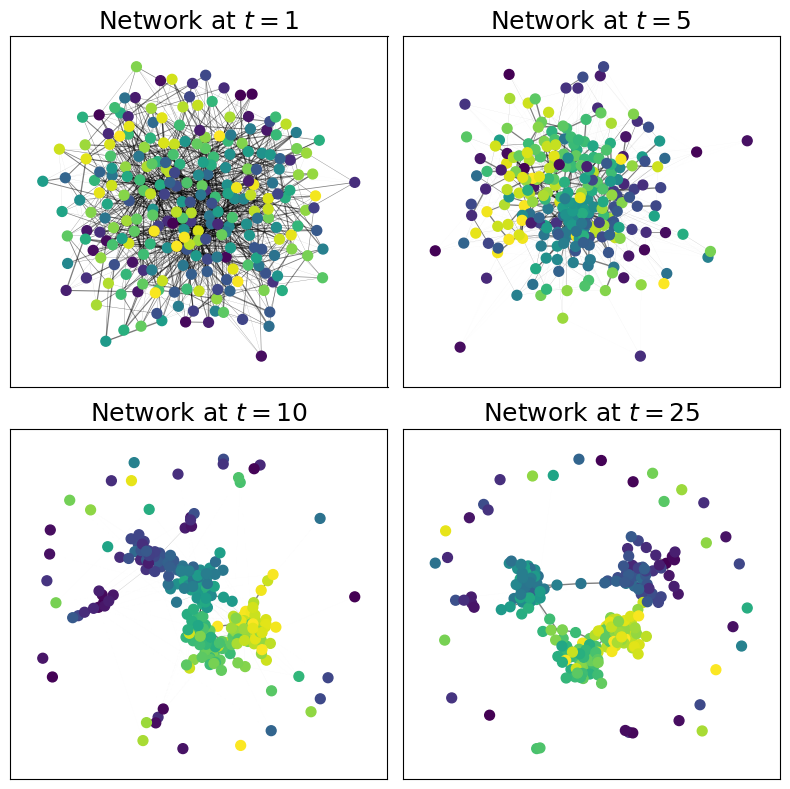}
        \caption{Snapshots of network structure. Edge width is drawn according to edge weight; nodes are coloured according to individuals' initial opinions. The network has a well-connected centre and fewer edges in the periphery. Over the dynamics, these outer edges are lost, leading to isolated individuals and three central clusters.}
    \end{subfigure}
    \caption{Demonstration of opinion formation under logistic weight dynamics using a \textbf{BA initial network}. A complete description of the setup can be found in Section \ref{Section: random network models}.}
    \label{fig:example BA network}
\end{figure}


\begin{figure}[ht!]
    \centering
    \begin{subfigure}{\textwidth}
        \centering
        \includegraphics[width=0.8\linewidth]{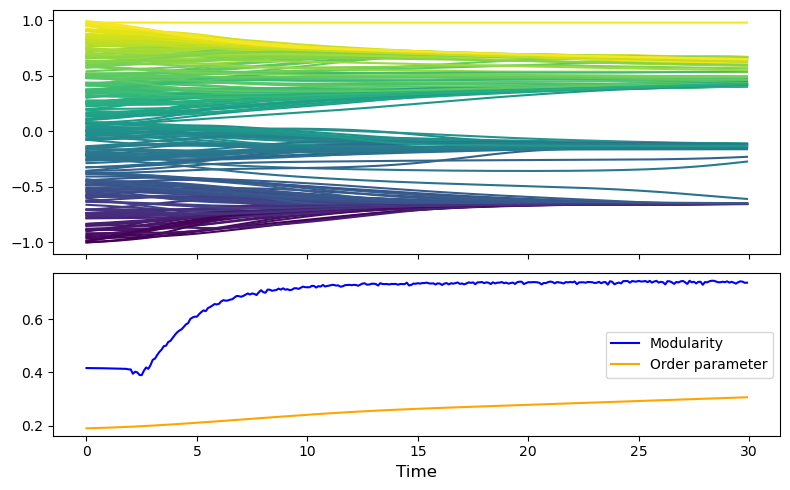}
        \caption{The top panel shows opinion trajectories over time, coloured according to initial opinion. Three distinct clusters form as weight dynamics rearrange the network's communities. The bottom panel shows the order parameter and network modularity over time. The modularity is initially constant, decreases slightly during the transition from two to three communities, then increases again.}
    \end{subfigure}
    \begin{subfigure}{\textwidth}
        \centering
        \includegraphics[width=0.75\linewidth]{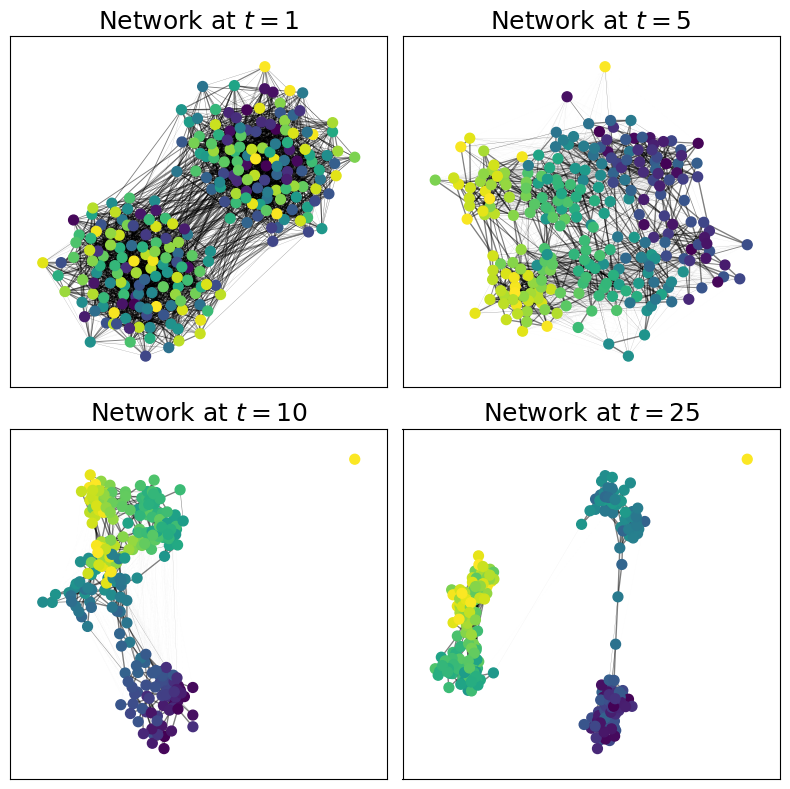}
        \caption{Snapshots of network structure. Edge width is drawn according to edge weight; nodes are coloured according to individuals' initial opinions. The network initially has two clear communities. Over time, these imposed communities are rearranged into three communities based on individuals' opinions.}
    \end{subfigure}
    \caption{Demonstration of opinion formation under logistic weight dynamics using a \textbf{SBM initial network}. A complete description of the setup can be found in Section \ref{Section: random network models}.}
    \label{fig:example SBM network}
\end{figure}

\end{document}